\documentclass[a4paper]{article}
\newif\ifchange

\changetrue

\usepackage{amsmath}
\usepackage{amssymb}
\usepackage{multirow}
\usepackage{bm}
\usepackage[normalem]{ulem}
\usepackage[mathlines]{lineno}
\usepackage{color}
\usepackage[dvipdfmx]{graphicx}

\ifchange

\else

\fi

\newcommand{\hide}[1]{}

\newtheorem{theorem}{Theorem}
\newtheorem{corollary}{Corollary}
\newtheorem{lemma}{Lemma}
\newtheorem{observation}{Observation}

\newenvironment{proof}{\noindent {\bf Proof.}}{\hfill$\Box$}

\title{Compatibility of convergence algorithms 
for autonomous mobile robots\thanks{An extended abstract of this paper
appeared in the Proc. International Colloquium on Structural Information 
and Communication Complexity (SIROCCO2023)~\cite{AY23}.}}

\author{Yuichi Asahiro and Masafumi Yamashita}

\date{May 27, 2023}

\begin{document}

\maketitle

\begin{abstract}
We investigate autonomous mobile robots in the Euclidean plane.
A robot has a function called {\em target function} to decide
the destination from the robots' positions.
Robots may have different target functions.
If the robots whose target functions are chosen from a set $\Phi$ 
of target functions always solve a problem $\Pi$,
we say that $\Phi$ is compatible with respect to $\Pi$.
If $\Phi$ is compatible with respect to $\Pi$,
every target function $\phi \in \Phi$ is an algorithm for $\Pi$.
Even if both $\phi$ and $\phi'$ are algorithms for $\Pi$,
$\{ \phi, \phi' \}$ may not be compatible with respect to $\Pi$.
From the view point of compatibility,
we investigate the convergence, the fault tolerant ($n,f$)-convergence (FC($f$)),
the fault tolerant ($n,f$)-convergence to $f$ points (FC($f$)-PO),
the fault tolerant ($n,f$)-convergence to a convex $f$-gon (FC($f$)-CP),
and the gathering problems, assuming crash failures.
Obtained results classify these problems into three groups:
The convergence, FC(1), FC(1)-PO, and FC($f$)-CP compose the first group:
{\bf Every} set of target functions which always shrink the convex hull of a configuration
is compatible.
The second group is composed of the gathering and FC($f$)-PO for $f \geq 2$:
{\bf No} set of target functions which always shrink the convex hull of a configuration
is compatible.
The third group, FC($f$) for $f \geq 2$, is placed in between.
Thus, FC(1) and FC(2),
FC(1)-PO and FC(2)-PO, 
and FC(2) and FC(2)-PO are respectively in different groups, 
despite that FC(1) and FC(1)-PO are in the first group.

\end{abstract}

\section{Introduction}

\subsection{Convergence problem and compatibility}

Over the last three decades,
swarms of autonomous mobile robots, e.g., automated guided vehicles and drones, 
have obtained much attention in a variety of contexts 
\cite{AP04,AOSY99,BA98,BFKPSW22,CFK97,FPSW99,KLOTW19,KIF93,MKK94,Parker93,PT00,SS96,SY99,WB96}.
Among them is understanding solvable problems by a swarm consisting of many 
simple and identical robots in a distributed manner,
which has been constantly attracting researchers in distributed computing society e.g.,
\cite{AP04,AOSY99,ASY22,BDT13,BFKPSW22,CP05,CP08,CDFH11,DFSY15,DPCMP16,DPT19,Flocchini19,FPS12,FPSW99,ISKI12,KLOTW19,Katreniak11,Prencipe19,SY99,Viglietta19,YS10,YUKY17}.

Many of the works mentioned above adopt the following robot model.
The robots look identical and indistinguishable.
Each robot is represented by a point that moves in the Euclidean plane.
It lacks identifier and communication devices,
and operates in Look-Compute-Move cycles.
When a robot starts a cycle,
it identifies the multiset of the robots' positions 
in its $x$-$y$ local coordinate system
(i.e., it has the strong multiplicity detection capability),
computes the destination point using a target function\footnote{
Roughly, a target function is a function from $(R^2)^n$ to $R^2$,
where $R$ is the set of real numbers and $n$ is the number of robots,
i.e., given a snapshot in $(R^2)^n$, it returns a destination point in $R^2$.
Later, we define a target function a bit more carefully.} 
based only on the multiset identified,
and then moves towards the destination point.
Here, the $x$-$y$ local coordinate system is right-handed (i.e., it has the chirality)
and its origin is always the position of the robot (i.e., it is self-centric),
and all robots are typically requested to take the same target function.

If each cycle starts at a time $t$ and finishes, reaching the destination,
before (not including) $t+1$, for some integer $t$, 
the scheduler is said to be {\em semi-synchronous} (${\cal SSYNC}$).
If cycles can start and end any time (even on the way to the destination),
it is {\em asynchronous} (${\cal ASYNC}$).

This paper investigates several {\em convergence problems},
e.g., \cite{AOSY99,CP05,CP08,CDFH11,DPCMP16,Flocchini19,FPS12,Katreniak11,SY99}.
The simplest convergence problem requires the robots to converge to a single point.
Under the ${\cal SSYNC}$ model,
the problem is solvable for robots with unlimited visibility \cite{SY99},
and is also solvable for robots with limited visibility \cite{AOSY99}.
Even under the ${\cal ASYNC}$ model, 
it is solvable by a target function called CoG,
which always outputs the center of gravity of the robots' positions  \cite{CP05}.
Finally, \cite{Katreniak11} gives a convergence algorithm,
under the ${\cal ASYNC}$ model, 
even for robots with limited visibility.

The authors of \cite{CP05} showed that CoG correctly works 
under the sudden-stop model,
under which the movement of a robot towards the center of gravity might stop 
on the way after traversing at least some fixed distance.
This implies that the robots can correctly converge to a point,
even if they take different target functions,
as long as they always move robots towards the current center of gravity
over distance at least some fixed constant.
This idea is extended in \cite{CDFH11}:
The authors proposed the $\delta$-{\em inner} property\footnote{
Let $P$, $D$, and $\bm{o}$ be the multiset of robots' positions,
the axes aligned minimum box containing $P$, and its center, respectively.
Define $\delta * D = \{ (1-2\delta) \bm{x} + 2\delta \bm{o} : \bm{x} \in D \}$.
A function $\phi$ is $\delta$-inner,
if $\phi(P) \in \delta * D$ for any $P$.}
of target functions,
and showed that the robot system converges to a point, 
if all robots take $\delta$-inner target functions, 
provided $\delta \in (0,1/2]$.

Consider a problem $\Pi$ and a set of target functions $\Phi$.
If the robots whose target functions are chosen from $\Phi$ always solve $\Pi$,
we say that $\Phi$ is {\em compatible with respect to} $\Pi$.
For example, every (non-empty) set of target functions satisfying 
the $\delta$-{\em inner} property is compatible with respect to 
the convergence problem under the ${\cal ASYNC}$ model \cite{CDFH11}.

If a singleton $\{ \phi \}$ is compatible with respect to $\Pi$,
we abuse to say that target function $\phi$ is an {\em algorithm}\footnote{
Here, we abuse term ``algorithm,''
since an algorithm must have a finite description.
A target function may not.
(See Section~\ref{SScontributions} for the definition of a target function.)
To compensate the abuse,
when to show the existence of a target function,
we will give a finite procedure to compute it.
To show its non-existence,
we will show the non-existence of a function (not only an algorithm).}
for $\Pi$.
If a set $\Phi$ of target functions is compatible with respect to $\Pi$, 
every target function $\phi \in \Phi$ is an algorithm for $\Pi$ by definition.
(The converse is not always true.)
Thus there is an algorithm for $\Pi$,
if and only if there is a compatible set $\Phi$ with respect to $\Pi$.
We sometimes say that a problem $\Pi$ is {\em solvable},
if there is a compatible set $\Phi$ with respect to $\Pi$,
which means that there is an algorithm for $\Pi$.

We would like to find a large compatible set $\Phi$ with respect to $\Pi$.
That $\Pi$ has a large compatible set with respect to $\Pi$ implies 
that $\Pi$ has many algorithms.
The difficulty of problems might be compared in terms of the sizes 
of their compatible sets.
A problem $\Pi$ which has a large compatible set $\Phi$ 
seems to have some practical merits, as well.
Two swarms both of which are controlled by target functions in $\Phi$
(which may be produced by different makers)
can merge to form a larger swarm, keeping the correctness of solving $\Pi$.
When a robot breaks down, 
we can safely replace it with another robot,
as long as it has a target function from $\Phi$.

\subsection{Convergence problems in the presence of crash faults}

This paper investigates three fault-tolerant convergence problems,
besides the convergence and the gathering problems.
This paper considers only crash faults:
A faulty robot can stop functioning at any time, becoming permanently inactive.
A faulty robot may not cause a malfunction, forever.
We cannot distinguish such a robot from non-faulty ones.
Let $n$ and $f (\leq n-1)$ be the number of robots and the number of faulty robots.

The {\em fault-tolerant (n,f)-convergence problem} (FC$(f)$) is the problem to find 
an algorithm which ensures that, as long as at most $f$ robots are faulty,
{\em all non-faulty robots} converge to a point.

The {\em fault-tolerant (n,f)-convergence problem to $f$ points} (FC$(f)$-PO) is the problem 
to find an algorithm which ensures that, as long as at most $f$ robots are faulty,
{\em all robots} (including faulty ones) converge to at most $f$ points.
All non-faulty robots need not converge to the same point.
If $f$ faulty robots have crashed at different positions,
each non-faulty robot must converge to one of the faulty robots.

The {\em fault-tolerant (n,f)-convergence problem to a convex f-gon} (FC$(f)$-CP) 
is the problem to find an algorithm which ensures that, 
as long as at most $f$ robots are faulty,
the convex hull of the positions of {\em all robots} (including faulty ones) converges 
to a convex $h$-gon $CH$ for some $h \leq f$, 
in such a way that, for each vertex of $CH$, there is a robot that converges to the vertex.

Since an algorithm for the FC(1)-PO solves the FC(1),
the former is not easier than the latter.
(Note that for $f \geq 2$, 
an algorithm for the FC($f$)-PO may not solve the FC($f$).)
Since an algorithm for the FC($f$)-PO solves the FC($f$)-CP,
again the former is not easier than the latter.
In \cite{CP05},
the authors showed that, for all $f \leq n-2$,
CoG is an algorithm for the FC($f$) under the ${\cal ASYNC}$ model.
To the best of the authors' knowledge,
the FC$(f)$-PO and the FC$(f)$-CP have not been investigated so far.

\subsection*{Gathering problems}

The {\em gathering problem} is similar to the convergence problem.
It requires the robots to gather in the exactly the same location.
The gathering problem has been investigated under a variety of assumptions,
e.g., \cite{AP04,BDT13,CFPS12,DFSY15,DPCMP16,Flocchini19,FPS12,SY99,YS10,YUKY17}.
Under the ${\cal SSYNC}$ model, 
the gathering problem is not solvable if $n = 2$.
If $n > 2$, it is solvable,
provided that all robots initially occupy distinct positions \cite{SY99}.
Under the ${\cal ASYNC}$ model, 
the same results hold \cite{CFPS12}.

Many other works investigate the gathering problem in the presence of crash faults.
See surveys \cite{DPCMP16,Flocchini19} for information on the fault-tolerant gathering problems.
The \emph{fault-tolerant $(n,f)$-gathering problem}, 
which is sometimes called the {\em weak gathering problem},
is the problem to find an algorithm  which ensures, as long as at most $f$ robots are faulty,
all non-faulty robots gather at a point.
In \cite{AP04}, the authors proposed a fault-tolerant $(n,1)$-gathering algorithm,
assuming that $n\geq 3$ and the robots initially occupy distinct positions.
Provided the chirality,
the fault-tolerant $(n,f)$-gathering problem can be solved for any $f < n$ 
except for the bivalent configuration\footnote{In the bivalent configuration, there are exactly two distinct points, where $n/2$ robots are located at each point.}~\cite{BDT13}.

Since the gathering problem is substantially harder than the convergence problem,
it may not be a good idea to use a gathering algorithm to solve the convergence problem,
especially when the size of solvable instances by a convergence algorithm is a major concern.

The gathering and convergence problems in the presence of Byzantine faults have also been
investigated or surveyed, e.g., in \cite{AP04,BPT09,BPT10,DPT19,Flocchini19}.

\subsection{Our contributions}
\label{SScontributions}

Let $R$ be the set of real numbers.
Formally, a {\em target function} $\phi$ is a function from $(R^2)^n$ 
to $R^2 \cup \{ \bot \}$ for all $n \geq 1$ such that $\phi(P) = \bot$, 
if and only if $(0,0) \not\in P$.
Here, $\bot$ is a special symbol to denote that $(0,0) \not\in P$. 
Suppose that a robot $r$ identifies a multiset $P$ of $n$ points,
which are the positions of the robots in its $x$-$y$ local coordinate system $Z$, 
in Look phase.
Then $(0,0) \in P$.\footnote{
Recall that $Z$ is self-centric;
that $(0,0) \not\in P$ means an error of eye sensor, 
which we assume will not occur, in this paper.}
Using its target function $\phi$, $r$ computes the target point 
$\bm{x} = \phi(P)$ in Compute phase.
Then it moves to $\bm{x}~(\not= \bot)$ in $Z$ in Move phase.

Let the convex hull and the center of gravity of $P$ be $CH(P)$ and $g(P)$, respectively.
For any $0 \leq d$, 
let $d * CH(P) = \{ d \bm{x} + (1-d)g(P) : \bm{x} \in CH(P)\}$.
The {\em scale} $\alpha(\phi)$ of a target function $\phi$ is defined by
\[
\alpha(\phi) = \sup_{P \in (R^2)^n} \alpha(\phi,P), 
\]
where $\alpha(\phi,P)$ is the smallest $d$ satisfying $\phi(P) \in d * (CH(P))$.\footnote{
For the sake of completeness, we assume that $\alpha(\phi,P) = 0$ when $\phi(P) = \bot$.}
Then the scale of a set $\Phi$ of target functions $\phi$ is defined by
\[
\alpha(\Phi) = \sup_{\phi \in \Phi} \alpha(\phi). 
\]
The only target function $\phi$ satisfying $\alpha(\phi) = 0$ is CoG.
Thus the set $\Phi$ of target functions satisfying $\alpha(\Phi) = 0$ is 
a singleton $\{ {\rm CoG} \}$.
The idea of scale is similar to that of $\delta$-inner property in \cite{CDFH11},
and more directly embodies the idea behind the $\delta$-inner target function,
in the sense that $\phi$ never expands and tends to shrink $CH(P)$ if $\alpha(\phi) < 1$.

Our contributions can be summarized in Table~\ref{Table0010}.

\begin{table}
\caption{The compatibility of a set $\Phi$ of target functions 
with respect to a problem $\Pi$, taking its scale $\alpha(\Phi)$ as a parameter.
Each entry contains the status A, E, N, or ? of the compatibility 
of $\Phi$ with respect to $\Pi$
(and the theorem/corollary/observation/citation number establishing the result
in parentheses).
Letter 'A' means that every $\Phi$ such that $\alpha(\Phi)$ is in the range is
compatible with respect to $\Pi$.
Letter 'N' means that any $\Phi$ such that $\alpha(\Phi)$ is in the range is 
not compatible with respect to $\Pi$,
which indicates the absence of an algorithm.
Letter 'E' means that some $\Phi$ is compatible,
while some other is not,
which indicates the existence of an algorithm.
Letter '?' means that the answer is unknown.}
\label{Table0010} 

\smallskip

\centering
\begin{tabular}{|c|l|l|l|}
\hline
 \multirow{2}{*}{problem $\Pi$} & \multicolumn{3}{c|}{scale $\alpha(\Phi)$}\\ \cline{2-4}
       & \multicolumn{1}{c|}{$\alpha(\Phi) = 0$} & \multicolumn{1}{c|}{$0 < \alpha(\Phi) < 1$} & \multicolumn{1}{c|}{$\alpha(\Phi) = 1$} \\ \hline \hline
  Convergence & A (Thm.~\ref{T0020}~\cite{CP05}) & A (Thm.~\ref{T0030}) & 
E (Thm.~\ref{T0040})  \\
  FC$(1)$    & A (\cite{CP05}) & A (Cor.~\ref{C1030}) & E (Thm.~\ref{T1020})\\
 FC$(1)$-PO & A (Thm.~\ref{T1030}) & A (Thm.~\ref{T1030}) & E (Thm.~\ref{T1020}) \\
 FC$(f)$-CP ($f \geq 2$) & A (Thm.~\ref{T3010}) & A (Thm.~\ref{T3010}) & E (Thm.~\ref{O3030}) \\
  FC$(f)$ ($f \geq 2$) & A (Thm.~\ref{T2010}~\cite{CP05}) & E (Thm.~\ref{T2020}) 
  & E (Cor.~\ref{C2025}) \\
  FC$(2)$-PO & N (Thm.~\ref{T4020}) & N (Thm.~\ref{T4020}) & E (Obs.~\ref{O4026}, Thm.~\ref{T4050})\\
  FC$(f)$-PO ($f \geq 3$) & N (Thm.~\ref{T4020}) & N (Thm.~\ref{T4020}) & ? \\\hline \hline
  Gathering  & N (Thm.~\ref{T5010}) & N (Thm.~\ref{T5010}) & E (Thm.~\ref{T5005}~\cite{SY99})  \\ 
\hline 
\end{tabular}

\end{table}

In Table~\ref{Table0010},
FC$(f)$, FC$(f)$-PO, and FC$(f)$-CP are respectively abbreviations of
the fault-tolerant $(n,f)$-convergence problem,
the fault-tolerant $(n,f)$-convergence problem to $f$ points, and
the fault-tolerant $(n,f)$-convergence problem to a convex $f$-gon. 
Letter 'A' in the entry of a problem $\Pi$ and a range of $\alpha(\Phi)$
means that every $\Phi$ such that $\alpha(\Phi)$ is in the range is 
compatible with respect to $\Pi$.
Letter 'N' means that any $\Phi$ such that $\alpha(\Phi)$ is in the range is 
not compatible with respect to $\Pi$,
which indicates the absence of an algorithm.
Letter 'E' means some $\Phi$ is compatible,
while some other is not,
which indicates the existence of an algorithm.
Letter '?' means that the answer is unknown.

For example, the entry of Convergence and $\alpha(\phi) = 0$ is A.
Thus $\{ {\rm CoG} \}$ is compatible with respect to the convergence problem,
which is equivalent to say that CoG is an algorithm for the convergence problem,
as \cite{CP05} shows.
Not only the case of $\alpha(\phi) = 0$,
but also the case of $0 < \alpha(\Phi) < 1$, 
every $\Phi$ is compatible with respect to the convergence problem.

\subsubsection*{Organization.}

After introducing the robot model we adopt in this paper in Section~\ref{Smodel},
we investigate the convergence problem, which has been studied extensively,
in Section~\ref{Sconvergence},
from the new viewpoint of understanding its compatibility.
In Section~\ref{Sonecrash},
we discuss the compatibilities of two convergence problems 
in the presence of at most one crash fault,
i.e., the FC(1) and the FC(1)-PO,
and show that they have the same property as the convergence problem.
Sections~\ref{SFCf} and \ref{SFCPf} respectively
investigate the compatibilities of the FC($f$) and the FC($f$)-CP for $f \geq 2$.
Every set $\Phi$ of target functions such that $0 \leq \alpha(\Phi) <1$ is
compatible with respect to the FC$(f)$-CP (like the FC(1), the FC(1)-PO, 
and the convergence problem),
while this property does not hold for the FC($f$) for $f \geq 2$.
Section~\ref{SFPOf} first shows that a target function $\phi$ is an algorithm 
for the FC($f$)-PO for $f \geq 2$, only if $\alpha(\phi) \geq 1$.
Thus any set $\Phi$ of target functions such that $0 \leq \alpha(\Phi) <1$ is 
{\bf not} compatible with respect to the FC($f$)-PO for $f \geq 2$, unlike the FC($f$).
We then present an algorithm $\psi_{(n,2)}$ for the FC(2)-PO.
Section~\ref{Sgather} investigates the gathering problem
to show the difference between this and the convergence problems 
from the viewpoint of compatibility.
We conclude the paper by presenting a list of open problems,
in Section~\ref{Sconclusion}.

\section{The model}
\label{Smodel}

Consider a robot system $\cal R$ consisting of $n$ robots $r_1, r_2, \ldots , r_n$.
Each robot $r_i$ has its own unit of length, and a local compass
defining an $x$-$y$ local coordinate system $Z_i$,
which is assumed to be right-handed and self-centric,
i.e., its origin $(0,0)$ is always the position of $r_i$.
We also assume that $r_i$ has the strong multiplicity detection capability,
i.e., it can count the number of robots resides at a point.

Given a target function $\phi_i$, 
each robot $r_i \in {\cal R}$ repeatedly executes a Look-Compute-Move cycle:
\begin{description}
 \item[Look:]  Robot $r_i$ identifies the multiset $P$ of the robots' positions 
(including the one of $r_i$) in $Z_i$.
Since $r_i$ has the strong multiplicity detection capability,
it can identify $P$ not only distinct positions of $P$.

\item[Compute:] Robot $r_i$ computes $\bm{x}_i = \phi_i(P)$.
(We do not mind even if $\phi_i$ is not computable.
We simply assume that $\phi_i(P)$ is given by an oracle.)

\item[Move:] Robot $r_i$ moves to $\bm{x}_i$.
We assume that $r_i$ always reaches $\bm{x}_i$ before this Move phase ends.
\end{description}

We assume a discrete time $0, 1, \ldots$.
At each time $t \geq 0$, 
the scheduler nondeterministically activates some subset 
(that may be none or all) of robots.
Then activated robots execute a cycle which starts at 
$t$ and ends before (not including) $t+1$,
i.e., the scheduler is semi-synchronous (${\cal SSYNC}$).

Let $Z_0$ be the 
$x$-$y$ global coordinate system,
which is right-handed and is not accessible by any robot $r_i$.
The coordinate transformation from $Z_i$ to $Z_0$ is denoted by $\gamma_i$.
We use $Z_0$ and $\gamma_i$ just for the purpose of explanation.

The position of robot $r_i$ at time $t$ in $Z_0$ is denoted by $\bm{x}_t(r_i)$.
Then $P_t = \{ \bm{x}_t(r_i) : 1 \leq i \leq n \}$ is a multiset representing
the positions of all robots at time $t$,
and is called the {\em configuration} at $t$.

Given an initial configuration $P_0$,
an assignment $\cal A$ of a target function $\phi_i$ to each robot $r_i$,
and an $\cal SSYNC$ schedule (produced by the $\cal SSYNC$ scheduler),
which decides the set of robots activated (to start a new Look-Compute-Move cycle) 
at each time instant $t$,
the execution of $\cal R$ is a sequence 
${\cal E}: P_0, P_1, \ldots , P_t, \ldots$ of configurations starting from $P_0$.
Here, for all $r_i$ and $t \geq 0$,
if $r_i$ is not activated at $t$, $\bm{x}_{t+1}(r_i) = \bm{x}_t(r_i)$.
Otherwise, if it is activated,
$r_i$ identifies $Q^{(i)}_t = \gamma^{-1}_i(P_t)$ in $Z_i$, 
computes $\bm{y} = \phi_i(Q^{(i)}_t)$, 
and moves to $\bm{y}$ in $Z_i$.\footnote{
Since $(0,0) \in Q^{(i)}_t$ by definition,
$\bm{y} \not= \bot$.}
Then $\bm{x}_{t+1}(r_i) = \gamma_i(\bm{y})$.
We assume that the scheduler is fair:
It activates every robot infinitely many times.
Throughout the paper, we regard the scheduler as an adversary.

Before closing this section,
we introduce several notations which we will use in the following sections.
Let $P \in (R^2)^n$.
The distinct points of $P$ is denoted by $\overline{P}$.
Then $|P|$ (resp. $|\overline{P}|$) denotes the number of points
(resp. the number of distinct points) in $P$.

Let $CH(P)$ be the convex hull of $P$.
While $P$ is a multiset of $n$ points,
$CH(P)$ is a convex region (including its inside).
We sometimes denote $CH(P)$ by a sequence of vertices of $CH(P)$ 
appearing on the boundary counter-clockwise.
Obviously, $CH(P) = CH(\overline{P})$.

The center of gravity $g(P)$ of $P$ is defined by 
$g(P) = \sum_{\bm{x} \in P} \bm{x}/n$.
Note that $g(P) \not= g(\overline{P})$ in general.

For two points $\bm{x}$ and $\bm{y}$ in $R^2$, 
$dist(\bm{x},\bm{y})$ denotes the Euclidean distance between $\bm{x}$ and $\bm{y}$.
For a set $B (\subseteq R^2)$ of points and a point $\bm{a} \in R^2$,
$dist(\bm{a},B) = \min_{\bm{x} \in B} dist(\bm{a},\bm{x})$.

Finally, let $\mathcal{P} = \{ P \in (R^2)^n: (0,0) \in P,\  n \geq 1 \}$.
Then $\mathcal{P}$ is the set of multisets of $n$ points that a robot may
identify in a Look phase.
We regard $\mathcal{P}$ as the domain of target functions.

\section{Convergence problem}
\label{Sconvergence}

We start our investigation with the convergence problem,
provided that all robots are non-faulty.
For any $0 \leq \alpha \leq 1$,
consider a target function CoG$_{\alpha}$ defined by
\[
{\rm CoG}_{\alpha}(P) = (1-\alpha)g(P),
\]
for any $P \in {\mathcal P}$.
The scale of CoG$_{\alpha}$ is $\alpha$,
and CoG$_0 =$ CoG.
The following theorem holds,
since CoG works correctly under the sudden stop model.

\begin{theorem}[\cite{CP05}]
\label{T0020}
For any $0 \leq \alpha < 1$,
let $\Phi_{\alpha} = \{ {\rm CoG}_{\alpha} \}$.
Then $\Phi_{\alpha}$  is compatible with respect to the convergence problem,
or equivalently, ${\rm CoG}_{\alpha}$ is an algorithm 
for the convergence problem.
\end{theorem}

We extend Theorem~\ref{T0020} to have the following theorem.
Let $N_{\epsilon}(B)$ be the $\epsilon$-neighbor of a set $B$,
i.e., $N_{\epsilon}(B) = \{ \bm{x}: dist(\bm{x},\bm{b}) < \epsilon, \bm{b} \in B \}$.
When $B$ is a singleton $\{ \bm{b} \}$,
we denote $N_{\epsilon}(B)$ by $N_{\epsilon}(\bm{b})$.

\medskip

\begin{theorem}
\label{T0030} 
Let $\Phi$ be a set of target functions such that $0 \leq \alpha(\Phi) < 1$.
Then $\Phi$ is compatible with respect to the convergence problem.
\end{theorem}

\begin{proof}
Let $\phi_i \in \Phi$ be the target function taken by robot $r_i$ 
for $i = 1, 2, \ldots , n$.
Let $\alpha(\phi_i) = \alpha_i$ and $\alpha = \max_{1 \leq i \leq n} \alpha_i$.
Then $\alpha \leq \alpha(\Phi) < 1$.

Consider any execution ${\cal E}: P_0, P_1, \ldots$ starting from 
any initial configuration $P_0$.
We show that $P_t$ converges\footnote{
Formally, you should read as ``a sequence $\{ P_t: t = 0, 1, \ldots \}$ converges.''
This convention applies in what follows.
Observe that if $P_t$ converges to a point,
then each robot converges to the point.
Later, when we discuss the convergence to multiple points,
the convergence of $P_t$ is not sufficient 
to show the convergence of each robot.}
to a point.

Suppose that $P_t = \{ \bm{x}, \bm{x}, \ldots , \bm{x} \}$ at some time $t$, 
i.e., $|\overline{P_t}|=1$.
Since $\bm{g}_t = g(P_t) = \bm{x}$, $P_{t+1} = P_t$.
Thus convergence has already been achieved.
We assume without loss of generality that $|\overline{P_t}| \geq 2$ for all $t \geq 0$.

Let $A_t \subseteq {\cal R}$ be the set of robots activated at time $t$.
If $\bm{x}_t(r) = \bm{g}_t$ for all $r \in A_t$, $P_{t+1} = P_t$ holds.
However, there is a robot $r$ such that $\bm{x}_t(r) \not= \bm{g}_t$ 
since $|\overline{P_t}| \geq 2$,
and $r$ is eventually activated by the fairness of scheduler.
Thus, without loss of generality, 
we assume that there is a robot $r \in A_t$ such that $\bm{x}_t(r) \not= \bm{g}_t$,
and that $P_{t+1} \not= P_t$ holds for all $t \geq 0$.\footnote{
Formally, let ${\cal E}'$ be a sequence constructed from $\cal E$
by removing all subsequences $P_{t+1}, P_{t+2}, \ldots , P_{t'}$ 
such that $P_{t-1} \not= P_t$,
$P_t = P_{t+1} = \ldots = P_{t'}$,
and $P_{t'} \not= P_{t'+1}$,
where we assume that $P_{-1} \not= P_0$.
Then ${\cal E}'$ is also an execution,
i.e., there is an activation schedule (that the $\cal SSYNC$ scheduler can produce)
that produces ${\cal E}'$,
and $\cal E$ converges if and only if ${\cal E}'$ does.
We consider ${\cal E}'$ instead of $\cal E$.
This is what this and similar assumptions in this paper formally mean.}

We denote $CH(P_t)$ by $CH_t$.
Since $\alpha < 1$,
$CH_{t+1} \subseteq CH_t$,
which implies that $CH_t$ converges to a convex $k$-gon $CH$
(including a point and a line segment) for some positive integer $k$.
We show that $CH$ is a point, i.e., $k = 1$.

Let $\bm{p}_0, \bm{p}_1, \ldots, \bm{p}_{k-1}$ be the vertices of $CH$ 
aligned counter-clockwise on the boundary.
To derive a contradiction,
we assume that $k \geq 2$.
For any pair $(i,j)~ (0 \leq i < j \leq k-1)$,
let $L_{(i,j)} = dist(\bm{p}_i, \bm{p}_j)$,
and $L = \min_{0 \leq i < j \leq k-1} L_{(i,j)}$.
Since $CH_t$ converges to $CH$,
for any $0 < \epsilon \ll (1-\alpha)L/n$,\footnote{
In the proof, all arguments below hold, e.g., for any $\epsilon$ such that
$0 < \epsilon < \frac{1-\alpha}{3-\alpha}\cdot \frac{L}{n}$. 
In what follows, like in this inequality, 
we use notation ``$\ll$ (much less than)'' or ``$\gg$ (much greater than)''
if deriving a bound is obvious and is not our concern.}
there is a time instant $t_0$ such that, for all $t \geq t_0$,
$CH \subseteq CH_t \subseteq N_{\epsilon}(CH)$.
Observe that for any vertex $\bm{p}$ of $CH$,
\[
dist(\bm{p},\alpha * CH_t) > (1 - \alpha)(L/n - \epsilon) - \epsilon \gg \epsilon,
\]
because $dist(\bm{p},\bm{g}_t) > L/n - \epsilon$.

Suppose that a robot $r$ is activated at some time $t \geq t_0$.
Then $\bm{x}_{t+1}(r) \in \alpha * CH_t$,
which implies that $\bm{x}_{t+1}(r) \not\in N_{\epsilon}(\bm{p})$,
for any vertex $\bm{p}$ of $CH$.
If $r$ is reactivated at some time $t' > t$ for the first time after $t$,
since $CH_{t'} \subseteq CH_t$ and $\bm{x}_{t'+1}(r) \in \alpha * CH_{t'}$,
$\bm{x}_{t'+1}(r) \not\in N_{\epsilon}(\bm{p})$,
for any vertex $\bm{p}$ of $CH$.
Therefore, for any $t' > t$ and any vertex $\bm{p}$ of $CH$,
$\bm{x}_{t'}(r) \not\in N_{\epsilon}(\bm{p})$.

On the other hand, 
all robots will be activated infinitely many times after time $t$,
by the fairness of scheduler.
It is a contradiction to the assumption that $CH_t$ converges to $CH$,
since there is a time instant $t' > t$ such that
for any robot $r$ and any vertex $\bm{p}$ of $CH$,
$\bm{x}_{t'}(r) \not\in N_{\epsilon}(\bm{p})$ holds.\footnote{
We showed a statement which is slightly stronger than 
what we needed to show here, to use this fact later.}
\end{proof}

\medskip

The following corollary holds by Theorem~\ref{T0030}.

\begin{corollary}
\label{C0020}
\begin{enumerate}
 \item 
Let $\phi$ be a target function such that $0 \leq \alpha(\phi) < 1$.
Then $\Phi = \{ \phi \}$  is compatible with respect to the convergence problem,
or equivalently, $\phi$ is a convergence algorithm.
\item
Let $\Phi$ and $\Phi'$ be two sets of target functions such that 
$0 \leq \alpha(\Phi) < 1$ and $0 \leq \alpha(\Phi') < 1$ hold.
Then $\Phi \cup \Phi'$ is also compatible with respect to the convergence problem,
not only $\Phi$ and $\Phi'$.
\end{enumerate}
\end{corollary}

Corollary~\ref{C0020} states that every target function $\phi$ 
such that $\alpha(\phi) < 1$ is a convergence algorithm.
However, some target function $\phi$ such that $\alpha(\phi) = 1$ is not 
a convergence algorithm.
Indeed, CoG$_1$ is an example as the following observation states.

\begin{observation}
\label{O0010} 
Target function \mbox{\rm CoG$_1$} is not a convergence algorithm.
Thus there is a set $\Phi$ of target functions such that $\alpha(\Phi) = 1$ 
and that it is not compatible with respect to the convergence problem.
\end{observation}

\begin{proof}
Consider any execution for two robots starting from configuration $P_0 = ((0,0),(1,0))$.
Since both robots take CoG$_1$ as their target functions
and CoG$_1$ does not move any robot in $P_0$,
it is not a convergence algorithm.
\end{proof}

\medskip

Let $\Phi$ and $\Phi'$ be two sets of target functions.
If $0 \leq \alpha(\Phi) < 1$ and $0 \leq \alpha(\Phi') < 1$,
$\Phi, \Phi'$, and $\Phi \cup \Phi'$ are all
compatible with respect to the convergence problem by Corollary~\ref{C0020}.
However, the following claim does {\bf not} hold:
\begin{quote}
If both of $\Phi$ and $\Phi'$ are compatible with respect to the convergence problem,
so is $\Phi \cup \Phi'$.
\end{quote}
To observe this fact, 
examine two target functions $\phi_T$ and $\phi_S$.
For a configuration $P$, define a condition $\Psi$ as follows:
\begin{description}
\item[$\Psi$:] 
$|P| = 7$, 
$(0,0) \in P$,
$P = T \cup S$, 
$T$ is an equilateral triangle,
$S$ is a square, 
$T$ and $S$ have the same side length, and
finally $T$ and $S$ do not overlap each other.
\end{description}

\medskip
\noindent
{\bf [Target function $\phi_T$]}
\begin{enumerate}
\item
If $P$ satisfies $\Psi$:
\begin{enumerate}
 \item 
If $(0,0) \in T$,
$\phi_T(P) = g(T)/2$,
which is the middle point on the line segment connecting $(0,0)$ and $g(T)$.
\item
If $(0,0) \in S$,
$\phi_T(P) = g(P)$. 
\end{enumerate}
\item
If $P$ does not satisfy $\Psi$:
$\phi_T(P) = g(P)$.
\end{enumerate}

\medskip
\noindent
{\bf [Target function $\phi_S$]}
\begin{enumerate}
\item
If $P$ satisfies $\Psi$:
\begin{enumerate}
 \item 
If $(0,0) \in S$,
$\phi_S(P) = g(S)/2$, 
which is the middle point on the line segment connecting $(0,0)$ and $g(S)$.
\item
If $(0,0) \in T$,
$\phi_S(P) = g(P)$. 
\end{enumerate}
\item
If $P$ does not satisfy $\Psi$:
$\phi_S(P) = g(P)$.
\end{enumerate}

Recall that $g(P)$, $g(T)$, and $g(S)$ are the centers of gravity 
of $P$, $T$, and $S$, respectively,
and that when a robot identifies $P$ in Look phase, $(0,0)$ always in $P$, 
which corresponds to its current position.
Let us observe that $\alpha(\phi_T) = 1$.
Since $\phi_T(P) \in CH(P)$ for all $P$, $\alpha(\phi_T) \leq 1$.
To see that $\alpha(\phi_T) \geq 1$,
consider any number $0 < a < 1$.
It is easy to construct a $P$ satisfying $\Psi$
such that $\frac{dist((0,0),g(T))}{dist((0,0)),g(P))} < a$,
which implies that $\alpha(\phi_T) > 1 - a$.
Thus $\alpha(\phi_T) = 1$ by the definition of $\alpha$.
By the same argument, $\alpha(\phi_S) = 1$.

Let $\Phi_T = \{ \phi_T \}, \Phi_S = \{ \phi_S \}$,
and $\Phi = \Phi_T \cup \Phi_S = \{ \phi_T, \phi_S \}$.
Then $\alpha(\Phi_T) = \alpha(\Phi_S) = \alpha(\Phi) = 1$.

\begin{theorem}
\label{T0040} 
Both $\Phi_T$ and $\Phi_S$ are compatible with respect to the convergence problem, 
but $\Phi = \Phi_T \cup \Phi_S$ is not.
\end{theorem}

\begin{proof}
(I) Let us start with showing that $\Phi_T$ is compatible with respect to the convergence problem.
A proof that $\Phi_S$ is also compatible with respect to the convergence problem is similar.

All robots take $\phi_T$ as their target functions.
Let ${\cal E}: P_0, P_1, \ldots$ be any execution starting from any initial configuration $P_0$.
When the number $n$ of robots is not $7$, $\phi_T$ = CoG. 
Thus $P_t$ converges to a point by Theorem~\ref{T0020}.

When $n = 7$,
$CH_{t+1} \subseteq CH_t$, and hence $CH_t$ converges to a convex $k$-gon 
for some positive integer $k$.
To derive a contradiction,
we first assume that $k = 2$.
That is, $CH$ is a line segment $\overline{\bm{pq}}$ connecting distinct points 
$\bm{p}$ and $\bm{q}$.
Let $L = dist(\bm{p},\bm{q})$.
For any $0 < \epsilon \ll L$,
there is a time $t_0$ such that, for any time $t \geq t_0$,
$CH \subseteq  CH_t \subseteq N_{\epsilon}(CH)$.

Suppose that $P_t$ satisfies $\Psi$.
In order for $P_{t+1}$ to satisfy $\Psi$,
no robots in $S$ are activated (to maintain a square),
and either all robots in $T$ are activated or none of them are activated
(to maintain a equilateral triangle).
Thus all robots in $T$ are activated since otherwise no robots are activated.
However, in this case, 
the sides of the triangle and the square are different in $P_{t+1}$.
Hence $P_{t+1}$ does not satisfy $\Phi$.

Suppose that $P_t$ does not satisfy $\Psi$.
If $P_{t'}$ does not satisfy $\Psi$ for all $t' \geq t$,
then $P_t$ converges to a point by Theorem~\ref{T0020},
which is a contradiction.

Thus, there is a $t'$ such that $P_{t'}$ does not satisfy $\Psi$,
but $P_{t'+1}$ does.
Since $P_{t'} \not= P_{t'+1}$ and $P_{t'}$ does not satisfy $\Psi$,
there is a robot $r$ who moves to $\bm{g}_{t'} = g(P_{t'})$ at $t'$.
In $P_{t'+1}$, 
$r$ at $\bm{g}_{t'}$ is a part of an equilateral triangle or a square
whose side has length at most $2 \epsilon$.
Also, since $CH \subseteq CH_{t'+1} \subseteq N_{\epsilon}(CH)$,
there are robots $r_{\bm{p}} \in N_{\epsilon}(\bm{p})$
and $r_{\bm{q}} \in N_{\epsilon}(\bm{q})$,
each of which is a part of an equilateral triangle or a square
whose side has length at most $2 \epsilon$,
since $P_{t'+1}$ satisfies $\Psi$.
However, it is a contradiction, 
since $\min \{ dist(\bm{p},\bm{g}_{t'}), dist(\bm{q}, \bm{g}_{t'}) \}$
$ > \frac{L}{7} - \epsilon \gg 4 \epsilon$.

When $k > 2$,
we can derive a contradiction using an argument similar to 
the proof of Theorem~\ref{T0030}.

\medskip
\noindent
(II) We show that $\Phi$ is not compatible with respect to the convergence problem.
Suppose that $P_0 = T \cup S$ satisfies $\Psi$,
and that the robots in $T$ (resp. $S$) take $\phi_T$ (resp. $\phi_S$) as
their target functions.
Suppose that the scheduler is ${\cal FSYNC}$,
i.e., all robots are activated at every time $t$.
(Note that the $\cal FSYNC$ scheduler always produce an $\cal SSYNC$ schedule.)
Then it is easy to observe that $P_t$ converges to $\{ g(T), g(S) \}$,
and does not converge to a point.
\end{proof}

\section{Convergence in the presence of at most one crash failure}
\label{Sonecrash}

We consider the convergence problem in the presence of at most one crash failure
in this section.
We investigate two problems, 
the fault-tolerant $(n,1)$-convergence problem (FC(1))
and the fault-tolerant $(n,1)$-convergence problem to a point (FC(1)-PO).
There is an algorithm for the FC(1)~\cite{CP05}.
Obviously, the FC(1)-PO is not easier than the FC(1),
since if all robots (including a faulty one) converge to a point,
then all non-faulty robots converge to a point.
We have the following theorem,
which implies that there is an algorithm for the FC(1)-PO.

\begin{theorem}
\label{T1030} 
Let $\Phi$ be a set of target functions,
and assume that $0 \leq \alpha(\Phi) < 1$.
Then $\Phi$ is compatible with respect to the \mbox{\rm FC(1)-PO}.
\end{theorem}

\begin{proof}
Since the proof is similar to that of Theorem~\ref{T0030},
we only give a sketch of the proof here.

Let $\phi_i \in \Phi$ be the target function taken by robot $r_i$ for $i = 1, 2, \ldots , n$.
Let $\alpha(\phi_i) = \alpha_i$ and $\alpha = \max_{1 \leq i \leq n} \alpha_i$.
Then $\alpha \leq \alpha(\Phi) < 1$.

Let ${\cal E}: P_0, P_1, \ldots$ be any execution starting from any initial configuration $P_0$.
Since $\alpha \leq \alpha(\Phi) < 1$,
$CH_{t+1} \subseteq CH_t$,
which implies that $CH_t$ converges to a convex $k$-gon $CH$ for some $k \geq 1$,
regardless of whether or not there is a faulty robot.

Like the proof of Theorem~\ref{T0030},
we assume that $k \geq 2$, 
and derive a contradiction.
Since $CH_t$ converges to $CH$,
for any small number $0 < \epsilon \ll (1-\alpha)L/n$,
there is a time $t_0$ such that for all $t > t_0$, 
$CH \subseteq CH_t \subseteq N_{\epsilon}(CH)$,
where $L$ is given in the proof of Theorem~\ref{T0030}.

The proof of Theorem~\ref{T0030} shows that 
there is a time $t' > t$ such that
for any non-faulty robot $r$ and any vertex $\bm{p}$ of $CH$,
$\bm{x}_{t'}(r) \not\in N_{\epsilon}(\bm{p})$ holds.

There is at most one faulty robot,
and $N_{\epsilon}(\bm{p}) \cap N_{\epsilon}(\bm{p}') = \emptyset$
for any two distinct vertices $\bm{p}$ and $\bm{p}'$ of $CH$.
Since $k \geq 2$,
there is a vertex $\bm{p}$ such that, 
for any robot $r$, $\bm{x}_{t'}(r) \not\in N_{\epsilon}(\bm{p})$,
which is a contradiction.
\end{proof}

\medskip

Immediately, we have the following corollary.

\begin{corollary}
\label{C1030} 
Let $\Phi$ be a set of target functions,
and assume that $0 \leq \alpha(\Phi) < 1$.
Then $\Phi$ is compatible with respect to the \mbox{\rm FC(1)}.
\end{corollary}

Next we reconsider the target functions $\phi_T$ and $\phi_S$
which we introduced in Section~\ref{Sconvergence}.

\begin{theorem}
\label{T1020} 
Both $\Phi_T = \{\phi_T\}$ and $\Phi_S = \{\phi_S\}$ are compatible 
with respect to the \mbox{\rm FC(1)-PO}.
However, $\Phi = \Phi_T \cup \Phi_S$ is not.
Recall that $\alpha(\Phi_T) = \alpha(\Phi_S) = \alpha(\Phi) = 1$.
\end{theorem}

\begin{proof}
(I) We show that $\Phi_T$ is compatible with respect to the FC(1)-PO.
A proof for $\Phi_S$ is similar.

If there is no faulty robot, 
all robots converge to a point by Theorem~\ref{T0040}.
Suppose that a robot $r_F$ is faulty at time $t$.
The proof is almost the same as the proof of Theorem~\ref{T0040},
so we just give a sketch of the proof.

Let ${\cal E}: P_0, P_1, \ldots$ be any execution starting 
from any initial configuration $P_0$.
Then $CH_t$ converges to a convex $k$-gon $CH$ for some positive integer $k$,
even if $\alpha(\Phi_T) = 1$ (since $CH_{t+1} \subseteq CH_t$ holds for all $t \geq 0$).

We assume $k = 2$ to derive a contradiction.
Let $CH = \overline{\bm{p}\bm{q}}$.
If $P_t$ satisfies $\Psi$ then $P_{t+1}$ does not,
provided $P_{t+1} \not= P_t$.
If $P_{t'}$ does not satisfy $\Psi$ for all $t' \geq t+1$, 
$P_t$ converges to a point by Theorem~\ref{T1030}.
Thus, there is a time $t' \geq t+1$ such that $P_{t'}$ does not satisfy $\Psi$,
but $P_{t'+1}$ does.
At $t'$, a robot $r (\not= r_F)$ moves to $\bm{g}_{t'} = g(P_{t'})$.
Now by the same argument as the proof of Theorem~\ref{T0040},
we can derive a contradiction to the assumption that $P_{t'+1}$ satisfies $\Psi$.

When $k > 2$, a contradiction can be derived by a similar argument.

\medskip
\noindent
(II) To observe that $\Phi$ is not compatible with respect to FC(1)-PO,
consider the case in which no faulty robot exists.
Then all robots need to converge to a point.
However, the execution of $\Phi$ converges to two points
when $P_0$ satisfies $\Psi$, the robots in $T$ take $\phi_T$, 
the robots in $S$ take $\phi_S$, 
and the scheduler is $\cal FSYNC$.

\end{proof}

\medskip

Observe that $\Phi$ is compatible with respect to the FC(1)-PO,
if a robot certainly crashes.

\section{FC($f$) for $f \geq 2$}
\label{SFCf}

We go on the fault tolerant $(n,f)$-convergence problem (FC$(f)$) for $f \geq 2$.
Since CoG is a fault tolerant $(n,f)$-convergence algorithm for all $f \leq n-2$ \cite{CP05},
and the problem is obviously solvable by CoG when exactly $f = n-1$ robots crash,
the next theorem holds.

\begin{theorem}[\cite{CP05}]
\label{T2010}
Suppose that $f \leq n-1$.
Set $\Phi_0 = \{ {\rm CoG} \}$ is compatible with respect to the \mbox{\rm FC$(f)$},
or equivalently,
a set $\Phi$ of target functions is compatible with respect to the \mbox{\rm FC$(f)$},
if $\alpha(\Phi) = 0$.
\end{theorem}

We show how the FC($f$) for $f \geq 2$ is different from the FC(1) (and the FC(1)-PO)
from the viewpoint of compatibility.
Corollary~\ref{C1030} states that every set $\Phi$ of target functions
such that $0 \leq \alpha(\Phi) < 1$ is compatible with respect to the FC(1).
We show that, for any $2 \leq f \leq n-1$ and $0 < \alpha < 1$,
there is a set $\Phi$ of two target functions such that (1) $\alpha(\Phi) = \alpha$,
(2) each target function in $\Phi$ is compatible with respect to the FC($f$), 
but (3) $\Phi$ is not compatible with respect to the FC($f$).

Consider first the following target function $\xi_{(\alpha,4)}$ for four robots,
where $0 < \alpha < 1$.
For a configuration $P$, define a condition $\Psi$ by a conjunction of two
conditions (i) and (ii):
\begin{description}
\item[$\Psi$:] 
\begin{description}
 \item[(i)] 
$P = \{ \bm{p}_1, \bm{p}_2, \bm{p}_3 , \bm{p}_4 \} \subseteq \overline{\bm{p}_1 \bm{p}_4}$,
where $\bm{p}_1, \bm{p}_2, \bm{p}_3$, and $\bm{p}_4$ are distinct 
and aligned on $\overline{\bm{p}_1 \bm{p}_4}$ in this order.
\item[(ii)]
$dist(\bm{p}_1,\bm{p}_2) = \frac{1}{2}L$ and
$dist(\bm{p}_2,\bm{p}_3) = \frac{2\alpha}{\alpha+3}L$,
where $L = dist(\bm{p}_1,\bm{p}_4)$.
\end{description}
\end{description}

\medskip
\noindent
{\bf [Target function $\xi_{(\alpha,4)}$]}
\begin{enumerate}
\item
Suppose that $P$ satisfies $\Psi$.
\begin{description}
\item[(a)] 
$\xi_{(\alpha,4)}(P) = \bm{p}_3$, if $\bm{p}_2 = (0,0)$,
\item[(b)]
$\xi_{(\alpha,4)}(P) = (0,0)$, if $\bm{p}_3 = (0,0)$, and
\item[(c)]
$\xi_{(\alpha,4)}(P) = g(P)$, otherwise.
\end{description}
\item
Suppose that $P$ does not satisfy $\Psi$.
Then $\xi_{(\alpha,4)}(P) = g(P)$.
\end{enumerate}

The following lemma holds.

\begin{lemma}
\label{L2010}
Set $\Phi = \{ \xi_{(\alpha,4)} \}$ is compatible with respect to 
the fault tolerant $(4,f)$-convergence problem,
where $0 < \alpha(\Phi) = \alpha < 1$ and $f \leq 3$.
\end{lemma}

\begin{proof}
By a simple calculation,
we have $\alpha(\xi_{(\alpha,4)}) = \alpha$,
and hence $\alpha(\Phi) = \alpha$.

Let ${\cal E}: P_0, P_1, \ldots$ be any execution starting from 
any initial configuration $P_0$.
If $CH_t$ is a convex $k$-gon with $k \geq 3$, 
then $\xi_{(\alpha,4)}(P_t) = {\rm CoG}(P_t)$.
Thus if there is no time $t$ such that $CH_t$ is a convex $k$-gon with $k \leq 2$,
all non-faulty robots converge to a point by Theorem~\ref{T2010}.
On the other hand, if $CH_t$ is a line segment for some time $t$,
then $CH_{t+1} \subseteq CH_t$ by the definition of $\xi_{(\alpha,4)}$,
which implies that $CH_{t'}$ is also a line segment for all $t' \geq t$.

We assume without loss of generality that $CH_0$ is a line segment.
Since $CH_{t+1} \subseteq CH_t$,
$CH_t$ converges to a line segment $CH = \overline{\bm{p}\bm{q}}$.
We denote the length of $CH_t$ (resp. $CH$) by $L_t$ (resp. $L$).
That is, $L_t$ is monotonically decreasing, i.e., $L_{t+1} \leq L_t$,
and converges to $L$.
Without loss of generality,
we may assume that all faulty robots have crashed at time $0$.

\medskip
\noindent
(I) Suppose that at most one robot is faulty.
Then $CH$ is a point,
and hence all (non-faulty) robots converge to the point.
To show this, 
we assume that $CH$ is not a point,
i.e., $\bm{p} \not= \bm{q}$ and $L > 0$, 
and derive a contradiction.

Since $CH_t$ converges to $CH$,
for any $0 < \epsilon \ll (1 - \alpha)L$,
there is a time $t_0 \geq 0$ such that, for all $t > t_0$,
$CH \subseteq CH_t \subseteq N_{\epsilon}(CH)$.
Let $P_t = \{ \bm{p}_1, \bm{p}_2, \bm{p}_3, \bm{p}_4 \}$.
Without loss of generality,
we assume that $\bm{p}_1, \bm{p}_2, \bm{p}_3$, and $\bm{p}_4$ are aligned 
in $\overline{\bm{p}_1\bm{p}_4}$ in this order,
where $\bm{p}_1 \in N_{\epsilon}(\bm{p})$ and $\bm{p}_4 \in N_{\epsilon}(\bm{q})$.

Suppose that $\bm{x}_t(r) = \bm{p}_1$ and $\bm{x}_t(r') = \bm{p}_4$.
Since either $r$ or $r'$ is non-faulty,
we assume that $r$ is non-faulty without loss of generality.
Since $r$ is non-faulty,
at some time $t' \geq t$,
it is activated and move to 
$g(P_{t'}) \not\in N_{\epsilon}(\bm{p}) \cup N_{\epsilon}(\bm{q})$.
Furthermore, it will never return to 
$N_{\epsilon}(\bm{p}) \cup N_{\epsilon}(\bm{q})$ again.
It is because, for all $t'' \geq t'$,
(i) $g(P_{t''}) \not\in N_{\epsilon}(\bm{p}) \cup N_{\epsilon}(\bm{q})$, and
(ii) if $P_{t''}$ satisfies $\Psi$, 
and $\bm{x}_{t''}(r)$ is one of the middle points of $CH_t$
($\bm{p}_2$ or $\bm{p}_3$ in the definition of $\xi_{(\alpha,4)}$),
the target point of $r$ ($\bm{p}_3$ in the definition of $\xi_{(\alpha,4)}$)
does not belong to $N_{\epsilon}(\bm{p}) \cup N_{\epsilon}(\bm{q})$,
since $\epsilon \ll (1 - \alpha)L$.

Thus all non-faulty robots eventually do not exist 
in $N_{\epsilon}(\bm{p}) \cup N_{\epsilon}(\bm{q})$,
which contradict to the assumption that $CH_t$ converges to $\overline{\bm{p}\bm{q}}$.

\medskip
\noindent
(II) Suppose next that at least two robots are faulty.
We show that all non-faulty robots converge to a point.
If more than one robot has crashed at the same position at time 0,
$P_t$ does not satisfy $\Psi$ for all $t > 0$,
and thus all non-faulty robots converge to a point by Theorem~\ref{T2010}. 

Consider the case in which all faulty robots have crashed at distinct positions.
Suppose that $CH_t$ converges to $CH = \overline{\bm{p}\bm{q}}$ for some
distinct points $\bm{p}$ and $\bm{q}$, since $CH$ cannot be a point.
By the same argument as in (I),
we can assume without loss of generality that
$\bm{p}$ and $\bm{q}$ are the positions of faulty robots at time $0$,
and that $P_0 = \{ \bm{p}, \bm{p}_2, \bm{p}_3, \bm{q} \}$,
where $\bm{p}, \bm{p}_2, \bm{p}_3$, and $\bm{q}$ are distinct 
and aligned on $\overline{\bm{p}\bm{q}}$ in this order,
that is, $CH_t = CH$ for all $t \geq 0$.

There are two robots $r$ and $r'$ such that $\bm{x}_0(r) = \bm{p}_2$ 
and $\bm{x}_0(r') = \bm{p}_3$,
one of which may be faulty.
If $P_t$ does not satisfy $\Psi$ for all $t \geq 0$, 
all non-faulty robots converge to a point by Theorem~\ref{T2010}.

Without loss of generality,
we assume that $P_0$ satisfies $\Psi$.
If $r$ is faulty, since $r'$ does not move by the definition of $\xi_{(\alpha,4)}$, 
$P_0 = P_t$ for all $t \geq 0$,
i.e., all non-faulty robots converges to $\bm{p}_3$.
Otherwise, if $r$ is non-faulty,
$r$ is eventually activated and moves to $\bm{p}_3$ at some time $t \geq 0$.

We observe that $P_{t'}$ does not satisfy $\Psi$ for all $t' > t$.
In $P_{t+1}$, $\bm{x}_{t+1}(r)$ and $\bm{x}_{t+1}(r')$ 
are located in $\overline{\bm{h}\bm{q}}$
(excluding $\bm{h}$), where $\bm{h} = (\bm{p}+\bm{q})/2$ 
is the middle point of $CH = \overline{\bm{p}\bm{q}}$.
Then $P_{t+1}$ does not satisfy $\Phi$.
Since $g(P_{t+1})$ is also in $\overline{\bm{h}\bm{q}}$ (excluding $\bm{h}$),
by a simple induction,
for all $t' > t$, $\bm{x}_{t'}(r)$ and $\bm{x}_{t'}(r')$ are located 
in $\overline{\bm{h}\bm{q}}$,
and hence $P_{t'}$ does not satisfy $\Psi$.
We can conclude that all non-faulty robots converge to a point by Theorem~\ref{T2010}.
\end{proof}

\medskip

For $0 < \alpha < 1$, 
we next consider the following target function $\xi'_{(\alpha,4)}$ for four robots.
We use the condition $\Psi$ defined above.

\medskip
\noindent
{\bf [Target function $\xi'_{(\alpha,4)}$]}
\begin{enumerate}
\item
Suppose that $P$ satisfies the condition $\Psi$.
\begin{description}
\item[(a)] 
$\xi'_{(\alpha,4)}(P) = (0,0)$, if $\bm{p}_2 = (0,0)$,
\item[(b)]
$\xi'_{(\alpha,4)}(P) = \alpha \bm{p}_1  +  (1 - \alpha) g(P)$, if $\bm{p}_3 = (0,0)$, and
\item[(c)]
$\xi'_{(\alpha,4)}(P) = g(P)$, otherwise.
\end{description}
\item
Suppose that $P$ does not satisfy $\Psi$.
Then $\xi'_{(\alpha,4)}(P) = g(P)$.
\end{enumerate}

By a similar argument, 
we can show that the set $\Phi' = \{ \xi'_{(\alpha,4)} \}$ is compatible 
with respect to the fault tolerant $(4,f)$-convergence problem,
and that $\alpha(\Phi') = \alpha(\xi'_{(\alpha,4)}) = \alpha < 1$.

\begin{corollary}
\label{C2010b}
Set $\Phi' = \{ \xi'_{(\alpha,4)} \}$ is compatible with respect to 
the fault tolerant $(4,f)$-convergence problem,
where $0 < \alpha(\Phi) = \alpha < 1$ and $f \leq 3$.
\end{corollary}

\begin{lemma}
\label{L2020} 
For any $2 \leq f \leq 3$,
$\Phi \cup \Phi' = \{ \xi_{(\alpha,4)}, \xi'_{(\alpha,4)}\}$ 
is not compatible with respect to the fault tolerant $(4,f)$-convergence problem.
Here, $\alpha(\Phi \cup \Phi') = \alpha(\Phi) = \alpha(\Phi') = \alpha < 1$.
\end{lemma}

\begin{proof}
Assume the followings:
$P_0$ satisfies the condition $\Psi$,
the target function of the robots at $\bm{p}_1, \bm{p}_3$, 
and $\bm{p}_4$ is $\xi_{(\alpha,4)}$,
that of the robot at $\bm{p}_2$ is $\xi'_{(\alpha,4)}$,
the scheduler is ${\cal FSYNC}$,
and the robots at $\bm{p}_1$ and $\bm{p}_4$ have already crashed at time $0$.

Then by the definitions of $\xi_{(\alpha,4)}$ and $\xi'_{(\alpha,4)}$,
$P_t = P_0$ for all $t \geq 0$.
Thus the lemma holds.
\end{proof}

\medskip

Let us extend this lemma to general $f$.
We use target functions $\xi_{(\alpha,n)}$ and $\xi'_{(\alpha,n)}$
which are easy extensions of $\xi_{(\alpha,4)}$ and $\xi'_{(\alpha,4)}$.
Let $\ell = \lfloor \frac{n-2}{2} \rfloor$ and
$\ell' = \lceil \frac{n-2}{2} \rceil$.
Thus $\ell + \ell' = n-2$.
For a configuration $P$, define a condition $\Psi^+$ by a conjunction of
two conditions $(i)$ and $(ii)$.
\begin{description}
\item[$\Psi^+$:] 
\begin{description}
\item[(i)] 
$P = \{ \bm{p}_1, \bm{p}_2, \ldots , \bm{p}_n \} \subseteq \overline{\bm{p}_1 \bm{p}_n}$,
where $\bm{p}_1, \bm{p}_{\ell+1}, \bm{p}_{\ell+2}, \bm{p}_{\ell+3}$ are distinct and 
aligned on $\overline{\bm{p}_1 \bm{p}_n}$ in this order, 
$\bm{p}_1 = \bm{p}_2 = \dots = \bm{p}_{\ell}$, i.e., the multiplicity of $\bm{p}_1$ is $\ell$, 
$\bm{p}_{\ell+3} = \bm{p}_{\ell+4} = \dots = \bm{p}_n$, 
i.e., the multiplicity of $\bm{p}_{\ell+3}$ is $\ell'$.
\item[(ii)]
Let $L = dist(\bm{p}_1,\bm{p}_n)$.
Then $dist(\bm{p}_1,\bm{p}_{\ell+1}) = \frac{1}{2}L$.
If $n$ is even,
$dist(\bm{p}_{\ell+1},\bm{p}_{\ell+2}) = \frac{\alpha n}{2(\alpha+n-1)}L$;
otherwise, if it is odd,
$dist(\bm{p}_{\ell+1},\bm{p}_{\ell+2}) = \frac{(2 \alpha -1)n+(1-\alpha)}{2(\alpha(n+1)-1)}L$.
\end{description}
\end{description}

\medskip
\noindent
{\bf [Target function $\xi_{(\alpha,n)}$]}
\begin{enumerate}
\item
Suppose that $P$ satisfies $\Psi^+$.
\begin{description}
\item[(a)] 
$\xi_{(\alpha,n)}(P) = \bm{p}_{\ell+2}$, if $\bm{p}_{\ell+1} = (0,0)$,
\item[(b)]
$\xi_{(\alpha,n)}(P) = (0,0)$, if $\bm{p}_{\ell+2} = (0,0)$, and
\item[(c)]
$\xi_{(\alpha,n)}(P) = g(P)$, otherwise.
\end{description}
\item
Suppose that $P$ does not satisfy $\Psi^+$.
Then $\xi_{(\alpha,n)}(P) = g(P)$.
\end{enumerate}

\medskip
\noindent
{\bf [Target function $\xi'_{(\alpha,n)}$]}
\begin{enumerate}
\item
Suppose that $P$ satisfies $\Psi^+$.
\begin{description}
\item[(a)] 
$\xi'_{(\alpha,n)}(P) = (0,0)$, if $\bm{p}_{\ell+1} = (0,0)$,
\item[(b)]
$\xi'_{(\alpha,n)}(P) = \alpha \bm{p}_1  +  (1-\alpha) g(P)$, 
if $\bm{p}_{\ell+2} = (0,0)$, and
\item[(c)]
$\xi'_{(\alpha,n)}(P) = g(P)$, otherwise.
\end{description}
\item
Suppose that $P$ does not satisfy $\Psi^+$
Then $\xi'_{(\alpha,n)}(P) = g(P)$.
\end{enumerate}

\begin{theorem}
\label{T2020}
For any $2 \leq f \leq n-1$ and $0 < \alpha < 1$,
there are two target functions $\xi_{(\alpha,n)}$ and $\xi'_{(\alpha,n)}$
such that (1) $\alpha(\xi_{(\alpha,n)}) = \alpha(\xi'_{(\alpha,n)}) = \alpha$,
(2) both of $\Phi = \{ \xi_{(\alpha,n)} \}$ and $\Phi' = \{ \xi'_{(\alpha,n)} \}$
are compatible with respect to the \mbox{FC$(f)$},
but (3) $\Phi \cup \Phi'$ is not.
\end{theorem}

\begin{proof}
It is not difficult to see that 
$\alpha(\xi_{(\alpha,n)}) = \alpha(\xi'_{(\alpha,n)}) = \alpha$. 
The theorem obviously holds from the proofs of Lemmas~\ref{L2010} and \ref{L2020}.
\end{proof}

\medskip

Recall Corollary~\ref{C1030}.
Theorem~\ref{T2020} states an interesting difference between the FC$(f)$ 
for $f \geq 2$ and the FC$(1)$.
Before closing this section,
we examine the case $\alpha = 1$ by using the above theorem.

\begin{corollary}
\label{C2025} 
For any $2 \leq f \leq n-1$,
there are two target functions $\xi_{(1,n)}$ and $\xi'_{(1,n)}$
such that (1) $\alpha(\xi_{(1,n)}) = \alpha(\xi'_{(1,n)}) = 1$,
(2) both of $\Phi = \{ \xi_{(1,n)} \}$ and $\Phi' = \{ \xi'_{(1,n)} \}$
are compatible with respect to the \mbox{\rm FC$(f)$},
but (3) $\Phi \cup \Phi'$ is not.
\end{corollary}

\begin{proof}
We construct target functions $\xi_{(1,n)}$ and $\xi'_{(1,n)}$ as follows:
First consider the following target function $\xi^*_{(1,3)}$ for three robots.

\medskip

\noindent
{\bf [Target function $\xi^*_{(1,3)}$]}
\begin{enumerate}
\item
If $P = \{ \bm{p}_1, \bm{p}_2, \bm{p}_3 \}$ satisfies 
(i) $P \subseteq \overline{\bm{p}_1 \bm{p}_3}$,
where $\bm{p}_1, \bm{p}_2, \bm{p}_3$ are distinct and 
aligned on $\overline{\bm{p}_1 \bm{p}_3}$ in this order, and
(ii) $dist(\bm{p}_1,\bm{p}_2) = 9L/10$ and $dist(\bm{p}_2,\bm{p}_3) = L/10$,
where $dist(\bm{p}_1,\bm{p}_3) = L$,
\begin{description}
\item[(a)] 
$\xi^*_{(1,3)}(P) = \bm{p}_3$, if $\bm{p}_1 = (0,0)$,
\item[(b)]
$\xi^*_{(1,3)}(P) = g(P)$, otherwise.
\end{description}
\item
Otherwise,
$\xi^*_{(1,3)}(P) = g(P)$.
\end{enumerate}

It is easy to observe that $\alpha(\xi^*_{(1,3)}) = 1$,
and that it is an algorithm for the fault tolerant $(3,2)$-convergence problem.

Let $\xi_{(1,n)}$ and $\xi'_{(1,n)}$ be respectively constructed from
$\xi_{(\alpha,n)}$ and $\xi'_{(\alpha,n)}$ by replacing
$\xi_{(\alpha,3)}$ and $\xi'_{(\alpha,3)}$ with $\xi^*_{(1,3)}$.
Then $\alpha(\xi_{(1,n)}) = \alpha(\xi'_{(1,n)}) = 1$.
By Theorem~\ref{T2020},
$\xi_{(1,n)}$ and $\xi'_{(1,n)}$ are algorithms for the FC($f$),
but $\{ \xi_{(1,n)}, \xi'_{(1,n)} \}$ is not compatible 
with respect to the FC($f$).
\end{proof}

\section{FC$(f)$-CP for $f \geq 2$}
\label{SFCPf} 

We next investigate the fault tolerant $(n,f)$-convergence problem
to a convex $f$-gon (FC$(f)$-CP).
The FC$(f)$-CP is the problem to ensure that, 
as long as at most $f$ robots crash,
in every execution ${\cal E}: P_0, P_1, \ldots$,
$CH_t = CH(P_t)$ converges to a convex $h$-gon $CH$ for some $h \leq f$,
in such a way that for each vertex of $CH$ there is a robot that
converges to the vertex.
A convex 2-gon is a line segment.
The FC(1)-CP is the FC(1)-PO.
The FC$(f)$-CP seems to be substantially easier than the FC$(f)$ (and the FC$(f)$-PO),
since the convergence of $CH_t$ to a convex $f$-gon does not always
mean the convergence of $P_t$;
a robot may not converge to a point.
We have the following theorem.

\begin{theorem}
\label{T3010} 
Let $\Phi$ be any set of target functions
such that $0 \leq \alpha(\Phi) < 1$.
Then $\Phi$ is compatible with respect to the \mbox{\rm FC($f$)-CP} 
for any $2 \leq f \leq n-1$.
\end{theorem} 

\begin{proof}
Let $\phi_i \in \Phi$ be the target function taken by robot $r_i$,
for $i = 1, 2, \ldots , n$.
Let $\alpha(\phi_i) = \alpha_i$ and $\alpha = \max_{1 \leq i \leq n} \alpha_i$.
Then $\alpha \leq \alpha(\Phi) < 1$.

Let ${\cal E}: P_0, P_1, \ldots$ be any execution of $\Phi$ 
starting from any initial configuration $P_0$.
We show that $CH_t$ converges to a convex $k$-gon $CH$ for some $k \leq f$,
provided that at most $f$ robots will crash.
Since $\alpha < 1$, $CH_t$ converges to a convex $k$-gon $CH$ for some $k\geq 1$.
We show that $k \leq f$.

To derive a contradiction,
we assume that $k \geq f+1$.
Let $\bm{p}_0, \bm{p}_1, \ldots , \bm{p}_{k-1}$ be the vertices of $CH$,
and they appear in this order on the boundary of $CH$ counter-clockwise.
For any pair $(i,j) (0 \leq i < j \leq k-1)$,
let $L_{(i,j)} = dist(\bm{p}_i,\bm{p}_j)$ and $L = \min_{0 \leq i < j \leq k-1} L_{(i,j)}$.
For any $0 < \epsilon \ll (1-\alpha)L/n$,
there is a time $t_0$ such that, for all $t > t_0$,
$CH \subseteq CH_t \subseteq N_{\epsilon}(CH)$.
By definition, for all $(i,j)$, 
$N_{\epsilon}(\bm{p}_i) \cap N_{\epsilon}(\bm{p}_j) = \emptyset$.

There is a time $t$ and a vertex $\bm{p}$
such that $N_{\epsilon}(\bm{p})$ does not include a faulty robot in $CH_t$.
Certainly such a $\bm{p}$ exists, 
since there are at most $f$ faulty robots,
a robot cannot belong to the neighbors of two vertices simultaneously,
and $k \geq f+1$.

Consider any non-faulty robot $r$ in $N_{\epsilon}(\bm{p})$ at $t$.
Then $r$ is eventually activated at time $t' > t$,
and moves inside $\alpha * CH_{t'} \subseteq \alpha * CH_t$.
Thus $dist(\bm{x}_{t'}(r), \bm{x}_{t'+1}(r)) > (1-\alpha) L/n \gg 2 \epsilon$,
which implies that there is no time $t'' > t'$
such that $\bm{x}_{t''}(r) \in N_{\epsilon}(\bm{p}_i)$ for any $0 \leq i \leq k-1$.
It contradicts to the assumption that $CH_t$ converges to $CH$,
since there is a time such that no robot is in $N_{\epsilon}(\bm{p})$.

We next show that for each vertex $\bm{p}$ of $CH$,
there is a robot $r$ that converges to $\bm{p}$.
By the argument above, if $r$ is non-faulty,
it eventually leaves $N_{\epsilon}(\bm{p})$ 
and will never return in $N_{\epsilon}(\bm{p}_i)$ for any $0 \leq i \leq k-1$.
Thus $r$ is faulty and it crashes at $\bm{p}$,
i.e., $r$ converges to $\bm{p}$.
\end{proof}

\begin{observation}
\label{O3025}
Let $\Phi$ and $\Phi'$ be any set of target functions
such that $\alpha(\Phi) < 1$ and $\alpha(\Phi') < 1$ hold.
Then all of $\Phi, \Phi'$, and $\Phi \cup \Phi'$ are compatible 
with respect to the \mbox{\rm FC($f$)-CP} for all $2 \leq f \leq n-1$,
since $\alpha(\Phi \cup \Phi') < 1$.
\end{observation}

However, we cannot extend Observation~\ref{O3025} 
to include the case $\alpha = 1$.
as we shall see below.

Consider the following two target functions $\tau$ and $\tau'$ for four robots.
For a configuration $P$, define a condition $\Psi$: 

\begin{description}
 \item[$\Psi$:] $P = \{ \bm{p}_1, \bm{p}_2, \bm{p}_3, \bm{p}_4\}$ is a convex quadrilateral 
such that $\angle \bm{p}_1 < \angle \bm{p}_2 < \angle \bm{p}_3 < \angle \bm{p}_4$,
where $\angle \bm{p}_i$ is the angle of vertex $\bm{p}_i$ of the quadrilateral.

\end{description}

\medskip
\noindent
{\bf [Target function $\tau$]}
\begin{enumerate}
\item If $|\overline{P}| \leq 2$, then $\tau(P) = (0,0)$.

\item If $P$ satisfies $\Psi$, then $\tau(P) = \bm{p}_1$.

\item Otherwise, $\tau(P) = g(P)$.
\end{enumerate}

\medskip
\noindent
{\bf [Target function $\tau'$]}
\begin{enumerate}
\item If $|\overline{P}| \leq 2$, then $\tau(P) = (0,0)$.

\item If $P$ satisfies $\Psi$, then $\tau'(P) = \bm{p}_4$.

\item Otherwise, $\tau'(P) = g(P)$.
\end{enumerate}

The following theorem holds.

\begin{theorem}
\label{O3030} 
Let $\Phi = \{ \tau \}$ and $\Phi' = \{ \tau' \}$.
Then $\alpha(\Phi) = \alpha(\Phi') = 1$.
Sets $\Phi$ and $\Phi'$ are compatible with respect to 
the fault tolerant $(4,2)$-convergence problem to a line segment,
but $\Phi \cup \Phi'$ is not.
\end{theorem}

\begin{proof}
Obviously, $\alpha(\tau) = \alpha(\tau') = 1$.
We first show that $\Phi$ is compatible with respect to 
the fault tolerant $(4,2)$-convergence problem 
to a convex 2-gon, i.e., to a line segment.  
A proof that $\Phi'$ is also compatible with respect to the problem is similar.  

Suppose that all robots take $\tau$ as their target functions, 
and let $\mathcal{E}: P_0, P_1, \dots$ be any execution 
starting from any initial configuration $P_0$.  
If $P_t$ does not satisfy $\Psi$ for all $t\geq 0$, then $\mathcal{E}$ converges to 
a line segment by Theorem~\ref{T3010}, since $\tau(P_t) = g(P_t)$ for all $t\geq 0$.

Without loss of generality, 
suppose that $P_0$ satisfies $\Psi$. 
If all robots activated at time $0$ are either a robot at $\bm{p}_1$,
or a robot which has been crashed, then $P_{1} = P_{0}$,
and hence, by the fairness of the scheduler,
a robot such that it is not at $\bm{p}_1$ and is not crashed  
is eventually activated.
Thus, without loss of generality,
we assume that at least one such robot is activated at time $0$.
(It is possible, since there are at most two faulty robots.)

If two or more such robots are activated,
then they move to $\bm{p}_1$, and $|\overline{P_{1}}| \leq 2$ holds.
Since $\tau(P_t) = (0,0)$ for $t\geq 1$, 
the execution converges to a line segment.

Suppose that exactly one such robot is activated.
Then, $CH(P_{1})$ is a triangle, 
in which two robots are at a vertex of $CH(P_{1})$.
By a simple induction on $t$,
for all $t \geq 1$,
$CH(P_t)$ is either a line segment or a triangle,
and does not satisfy $\Psi$.
Thus, $\mathcal{E}$ converges to a line segment, again by Theorem~\ref{T3010}.

Next, we show that $\Phi \cup \Phi'$ is not compatible with respect to 
the fault tolerant $(4,2)$-convergence problem to a line segment.
Suppose that $P_0$ satisfies $\Psi$.
Consider the case that 
the robot at $\bm{p}_1$ (resp. $\bm{p}_4$) takes $\tau$ (resp. $\tau'$) as its target function, 
and the robots at $\bm{p}_2$ and $\bm{p}_3$ crash at time 0.
Since $\tau(P_0) = (0,0)$ and $\tau'(P_0) = (0,0)$, $P_t = P_0$ for all $t\geq 1$, i.e., 
the execution does not converge to a line segment.
\end{proof}

\section{FC$(f)$-PO for $f \geq 2$}
\label{SFPOf} 

This section investigates the fault tolerant $(n,f)$-convergence problem
to $f$ points (FC$(f)$-PO) for $f \geq 2$.
At a glance, the FC$(f)$-PO looks to have properties similar to the FC$(f)$,
and readers might consider that the former would be easier than the latter,
since in the former, all non-faulty robots are not requested to converge 
to a single point.
On the contrary, 
we shall see that the FC$(f)$-PO is a formidable problem even if $f = 2$.

\subsection{Compatibility}
\label{SSCompatibility}

We show a difference between the FC$(f)$ and the FC$(f)$-PO for $f \geq 2$,
from the viewpoint of compatibility.

\begin{theorem}
\label{T4020} 
Let $f \geq 2$.
Any target function $\phi$ is not an algorithm for the \mbox{\rm FC$(f)$-PO},
if $0 \leq \alpha(\phi) < 1$,
or equivalently, $\Phi$ is not compatible with the \mbox{\rm FC$(f)$-PO},
if $0 \leq \alpha(\Phi) < 1$.
\end{theorem}

\begin{proof}
To derive a contradiction,
we assume that there is an algorithm $\phi$ for the FC$(f)$-PO 
such that $\alpha(\phi) < 1$.

Let ${\cal R} = \{ r_1, r_2, \ldots , r_n \}$, where $n \geq f+1$.
Consider an initial configuration $P_0$
such that points $\bm{x}_0(r_1), \bm{x}_0(r_2), \ldots , \bm{x}_0(r_f)$ 
form a regular $f$-gon $B$,
and $\bm{x}_0(r_i)$ is the center of gravity of $B$,
where $i = f+1, f+2, \ldots , n$.

Consider any execution ${\cal E}: P_0, P_1, \ldots$,
assuming that the robots $r_1, r_2, \ldots, r_f$ have already crashed at time $0$.
Then $r_{f+1}$ must converge to one of the vertices of $B$.
Since $CH_t = B$ for all $t \geq 0$,
it is a contradiction since $\alpha(\phi) < 1$.
\end{proof}

\medskip

Recall that $\Phi = \{ \xi_{(\alpha, n)}\}$ 
and $\Phi' = \{\xi'_{(\alpha, n)}\}$ are compatible with respect to 
the FC$(f)$ for all $2\leq f \leq n-1$ and $0 \leq \alpha < 1$ by Theorem~\ref{T2020}.
Since $\alpha(\Phi) = \alpha$,
by Theorem~\ref{T4020}, we have:

\begin{corollary}
\label{O4010} 
Neither $\Phi$ nor $\Phi'$ is compatible with respect to the \mbox{\rm FC$(f)$-PO},
for all $f \geq 2$ and $0 \leq \alpha < 1$.
\end{corollary}

It goes without saying that not all target functions $\phi$ with $\alpha(\phi) = 1$ 
are algorithms for the FC$(2)$-PO. 
For example,

\begin{observation}
\label{O4026} 
CoG$_1$ is not an algorithm for the \mbox{\rm FC$(2)$-PO}.
\end{observation}

\subsection{Algorithm for FC$(2)$-PO}
\label{SSFC2PO}

In Section~\ref{SSCompatibility},
we showed that, for any $f \geq 2$, 
there is no FC($f$)-PO algorithm whose scale is less than 1.
It is a clear difference between the FC$(f)$-PO and the FC($f$),
which is solved, e.g., by CoG$_{\alpha}$ for any $0 \leq \alpha < 1$.
This section proposes an algorithm $\psi_{(n,2)}$ with $\alpha(\psi_{(n,2)})=1$
for the FC$(2)$-PO, and shows its correctness.
Unfortunately, proposing an algorithm for FC($f$)-PO for an $f \geq 3$ 
is left as a future work.

\subsubsection{Algorithm $\psi_{(3,2)}$}
\label{SSSA32}

We start with proposing an algorithm $\psi_{(3,2)}$ for solving the
fault tolerant $(3,2)$-convergence problem to two points.
Let $>$ be a lexicographic order on $R^2$ defined as follows:
For two distinct points $\bm{p} = (p_x,p_y)$ and $\bm{q} = (q_x,q_y)$ in $R^2$,
$\bm{p} > \bm{q}$ if and only if
either (i) $p_x > q_x$, or (ii) $p_x = q_x$ and $p_y > q_y$ holds.\footnote{
We use the same notation $>$ to denote the lexicographic order on $R^2$
and the order $>$ on $R$ to save the number of notations.}

We classify configurations $P = \{ \bm{p}_1, \bm{p}_2, \bm{p}_3 \}$ 
such that $(0,0) \in P$ into three types G, L, and T:

\begin{description}
\item[G(oal):] 
There are $i,j (i \not= j)$ such that $\bm{p}_i = \bm{p}_j$.

\item[L(ine):]
Three points $\bm{p}_1, \bm{p}_2, \bm{p}_3$ are distinct,
and $\bm{p}_2 \in \overline{\bm{p}_1\bm{p}_3}$.
We assume $\bm{p}_1 > \bm{p}_3$. 

\item[T(riangle):]
$P$ forms a triangle (not including a line and a point).
We assume that $\bm{p}_1, \bm{p}_2, \bm{p}_3$ appear in this order counter-clockwise
on the boundary of the triangle.
\end{description}

We describe target function $\psi_{(3,2)}$.

\medskip
\noindent
{\bf [Target function $\psi_{(3,2)}$]}
\begin{enumerate}
\item
If $P$ is type G, $\psi_{(3,2)}(P) = (0,0)$.

\item
When $P$ is type L:
\begin{description}
 \item[(a)] 
If $\bm{p}_1 = (0,0)$ or $\bm{p}_3 = (0,0)$,
then $\psi_{(3,2)}(P) = \bm{p}_2/2$.
\item[(b)]
If $\bm{p}_2 = (0,0)$, $\psi_{(3,2)}(P) = \bm{p}_1/2$.
\end{description}

\item
When $P$ is type T:
If $\bm{p}_i = (0,0)$,
then $\psi_{(3,2)}(P) = \bm{p}_{i+1}/2$,
where $\bm{p}_4 = \bm{p}_1$.
\end{enumerate}

We show the correctness of $\psi_{(3,2)}$.
Let ${\cal R} = \{ r_1, r_2, r_3\}$.
Recall that the $x$-$y$ local coordinate system $Z_i$ of each robot $r_i$ 
is right-handed.

\begin{lemma}
\label{L4030} 
Target function $\psi_{(3,2)}$ satisfies $\alpha(\psi_{(3,2)}) = 1$,
and is an \mbox{\rm FC(2)-PO} algorithm for $n = 3$.
\end{lemma}

\begin{proof}
We observe $\alpha(\psi_{(3,2)}) = 1$.
Since $\psi_{(3,2)}(P) \in CH(P)$, 
$\alpha(\psi_{(3,2)}) \leq 1$.
For any $0 < a < 1$,
there is a configuration $P$ of type L such that $\bm{p}_1 > \bm{p}_3$, 
$\bm{p}_2 = (0,0)$, $dist(\bm{p}_1,\bm{p}_3) = 1$, 
and $dist(\bm{p}_1,\bm{p}_2) = a$ hold.
Then $dist(\bm{p}_1,g(P)) = (1+a)/3$,
and $\alpha(\psi_{(3,2)},P) = 1 - 3a/(2(1+a))$.
By definition,
$\alpha(\psi_{(3,2)}) \geq \lim_{a \rightarrow 0} (1 - 3a/(2(1+a))) = 1$.
Thus $\alpha(\psi_{(3,2)}) = 1$.

Let ${\cal E}: P_0, P_1, \ldots$ be any execution starting 
from any initial configuration $P_0$.
We show that each robot converges to one of at most two convergence points.

Let $Q_t^{(i)}$ be the multiset of the robots' positions
that $r_i$ identifies in Look phase in $Z_i$ at time $t$.
That is, $\gamma_i(Q_t^{(i)}) = P_t$,
where $\gamma_i$ is the coordinate transformation from $Z_i$ to $Z_0$.
Of course, $Q_t^{(i)}$ and $P_t$ are similar,
and have the same type.
Keep in mind that $r_i$ computes $\psi_{(3,2)}(Q_t^{(i)})$,
not $\psi_{(3,2)}(P_t)$,
and $Q_t^{(i)} \not= Q_t^{(j)}$ occurs in general.

If $P_t$ (and $Q_t^{(i)}$) is type G, 
$\psi_{(3,2)}(Q_t^{(i)}) = (0,0)$ for all robot $r_i$,
which means that $r_i$ does not move,
since $(0,0)$ is the current position of $r_i$ in $Z_i$.
Since $|\overline{P_t}| \leq 2$ by the definition of type G,
and each robot $r_i$ does not move after $t$
(regardless of whether or not it is faulty), 
it converges to one of at most two convergence points.

Suppose that $P_t$ (and $Q_t^{(i)}$) is type L.
Without loss of generality,
we assume $P_t = \{ \bm{x}_t(r_1), \bm{x}_t(r_2), \bm{x}_t(r_3) \}$,
where $\bm{x}_t(r_1), \bm{x}_t(r_2), \bm{x}_t(r_3)$ are distinct,
and $\bm{x}_t(r_2) \in \overline{\bm{x}_t(r_1)\bm{x}_t(r_3)}$.

Suppose that $\gamma^{-1}_i(\bm{x}_t(r_j)) = \bm{y}^{(i)}_t(r_j)$,
for all $i,j = 1,2,3$.
That is, $Q^{(i)}_t = \{ \bm{y}^{(i)}_t(r_1), \bm{y}^{(i)}_t(r_2), \bm{y}^{(i)}_t(r_3) \}$,
where $\bm{y}^{(i)}_t(r_i) = (0,0)$.

The target position of $r_1$ (resp. $r_3$) is the middle point of
$\bm{x}_t(r_1)$ (resp. $\bm{x}_t(r_3)$) and $\bm{x}_t(r_2)$,
and that of $r_2$ is the middle point of $\bm{x}_t(r_2)$ 
and $\bm{x}_t(r_j)$, where $j$ is either 1 or 3,
and is determined from $Q^{(2)}_t$ (not from $P_t$),
since $r_2$ computes $\psi_{(3,2)}(Q^{(2)}_t)$ (not $\psi_{(3,2)}(P_t)$) in Compute phase;
$j = 1$ if $\bm{y}^{(2)}_t(r_1) > \bm{y}^{(2)}_t(r_3)$,
and otherwise, $j = 3$.
(Since $P_t$ is type L (but not G), 
$\bm{y}^{(2)}_t(r_1) = \bm{y}^{(2)}_t(r_3)$ does not occur.)

We observe the following:
$P_{t'}$ is type L for all $t' \geq t$.
Moreover, $\bm{x}_{t'}(r_2) \in \overline{\bm{x}_{t'}(r_1)\bm{x}_{t'}(r_3)}$,
i.e., the relative positions of $r_1, r_2$, and $r_3$ do not change,
unless some of them match and the type changes to $G$.
Thus if $\bm{y}^{(2)}_t(r_1) > \bm{y}^{(2)}_t(r_3)$,
then $\bm{y}^{(2)}_{t'}(r_1) > \bm{y}^{(2)}_{t'}(r_3)$ for all $t' > t$.
Therefore, if the target position of $r_2$ at time $t$ is
the middle point of $\bm{x}_t(r_2)$ and $\bm{x}_t(r_j)$,
so is at time $t'$ for all $t' > t$.
Thus $r_2$ converges either to $r_1$ or to $r_3$,
as long as $r_2$ is not faulty.

If $r_2$ is faulty,
either $r_1$ or $r_3$ is not faulty,
and it converges to $r_2$.

Suppose that $P_t$ is type T.
For $i = 1,2,3$, 
let $Q^{(i)}_t = \{ \bm{q}^{(i)}_1, \bm{q}^{(i)}_2, \bm{q}^{(i)}_3 \}$,
where we assume that $\bm{q}^{(i)}_1, \bm{q}^{(i)}_2, \bm{q}^{(i)}_3$ appear 
in this order counter-clockwise on the boundary of the triangle.\footnote{
$\bm{q}^{(i)}_j$ may not be the position $\bm{y}^{(i)}_t(r_j)$ of robot $r_j$ 
at time $t$ observed by $r_i$ in $Z_i$.}
Since $Z_i$ is right-handed, 
once activated,
a robot $r_i$ recognizes $j$ such that $\bm{q}^{(i)}_j = (0,0)$,
and moves to the middle point of the current position $(0,0)$
and the next point $\bm{q}^{(i)}_{j+1}$ in $Z_i$,
which is, in $Z_0$, the middle point of the current position of $r_i$ 
and the next point counter-clockwise in $P_t$.
Thus $P_{t'}$ is type T for all $t' \geq t$.

Now, it is obvious that
(1) if there is no faulty robot,
then all robots converge to a point,
(2) if exactly one robot crashes at a point $\bm{p}$,
then the other two robots converge to $\bm{p}$,
(3) if exactly two robots crash at points $\bm{p}$ and $\bm{p}'$,
then the third robot converges either to $\bm{p}$ or $\bm{p}'$.
\end{proof}

\subsubsection{Algorithm $\psi_{(n,2)}$}
\label{SSSAn2}

We propose an algorithm $\psi_{(n,2)}$ to solve the FC(2)-PO for $n \geq 4$,
and show its correctness.
By combining $\psi_{(3,2)}$ and $\psi_{(n,2)}$,
we obtain an algorithm for the FC$(2)$-PO.

Let LN$_{(n,2)}$ (for $n \geq 4$) be an algorithm to solve the FC(2)-PO
in such a way that $CH(P_t) \subseteq CH(P_0)$ for all $t \geq 0$,
provided that $CH(P_0)$ is a line segment.
Assuming the existence of LN$_{(n,2)}$,
we propose $\psi_{(n,2)}$ and show its correctness.
Later, we propose LN$_{(n,2)}$ and show its correctness.

In $\psi_{(3,2)}$, 
we used lexicographic order $>$ to compare positions $\bm{p}$ and $\bm{q}$.
Let $\sqsubset$ be a lexicographic order on $(R^2)^n$ defined as follows:
For distinct multisets of $n$ points
$P = \{ \bm{p}_1, \bm{p}_2, \ldots, \bm{p}_n \}$
and $Q = \{ \bm{q}_1, \bm{q}_2, \ldots, \bm{q}_n \}$,
where for all $i = 1, 2, \ldots , n-1$, 
$\bm{p}_i \leq \bm{p}_{i+1}$ and $\bm{q}_i \leq \bm{q}_{i+1}$ hold,
$P \sqsubset Q$, if and only if there is an $i (1 \leq i \leq n-1)$
such that (i) $\bm{p}_j = \bm{q}_j$ for all $j = 1, 2, \ldots, i-1$,\footnote{
We assume $\bm{p}_0 = \bm{q}_0$.}
and (ii) $\bm{p}_i < \bm{q}_i$.

Although these orders are easy to understand,
they have a drawback for our purpose,
since robots cannot consistently compute
lexicographic orders $<$ and $\sqsubset$.
To see this fact,
let $\bm{x}$ and $\bm{y}$ be distinct points in $\overline{P}$ in $Z_0$.
Then both $\gamma^{-1}_i(\bm{x}) < \gamma^{-1}_i(\bm{y})$
and $\gamma^{-1}_i(\bm{x}) > \gamma^{-1}_i(\bm{y})$ can occur,
depending on $Z_i$.
Thus robots cannot consistently compare $\bm{x}$ and $\bm{y}$ using $>$.
And it is true for $\sqsubset$, as well.
In $\psi_{(n,2)}$, 
we introduce and use an order $\succ$ that all robots can consistently compute.

Let $P = \{ \bm{p}_1, \bm{p}_2, \ldots , \bm{p}_n \}$ be a multiset of $n$ points,
and $\overline{P} = \{\bm{q}_1, \bm{q}_2, \ldots , \bm{q}_m \}$ 
be the set of distinct points in $P$.
We denote the multiplicity of $\bm{q}$ in $P$ by $\mu_P(\bm{q})$,
i.e., $\mu_P(\bm{q}) = |\{ i: \bm{p}_i = \bm{q} \in P \}|$.
We identify $P$ with a pair $(\overline{P},\mu_P)$,
where $\mu_P$ is a labeling function to associate 
label $\mu_P(\bm{q})$ with each element $\bm{q} \in \overline{P}$.
Let $G_P$ be the rotation group $G_{\overline{P}}$ of $\overline{P}$ 
about $\bm{o}_P$ preserving $\mu_P$,
where $\bm{o}_P$ is the center of the smallest enclosing circle of $P$.
The order $|G_P|$ of $G_P$ is denoted by $k_P$.
We define $k_P = 0$, 
if $|\overline{P}| = 1$, i.e., if $\overline{P} = \{ \bm{o}_P \}$.\footnote{
The symmetricity of $P$ is $\sigma(P) = GCD(k_P, \mu_P(\bm{o}_P))$ \cite{SY99}.}

For example,
let $P_1 = \{ \bm{a}, \bm{b}, \bm{c} \}$, $P_2 = \{ \bm{a}, \bm{b}, \bm{c}, \bm{c}\}$, 
$P_3 = \{ \bm{a}, \bm{a}, \bm{b}, \bm{b}, \bm{c}, \bm{c} \}$,
and $P_4 = \{ \bm{a}, \bm{b}, \bm{c}, \bm{o}, \bm{o} \}$, 
where a triangle $\bm{abc}$ is equilateral,
and $\bm{o}$ is the center of the smallest enclosing circle of triangle $\bm{abc}$.
Then $k_{P_1} = k_{P_3} = k_{P_4} = 3$ and $k_{P_2} = 1$.

Suppose that $P$ is a configuration in $Z_0$.
When activated, a robot $r_i$ identifies the robots' positions 
$Q^{(i)} = \gamma^{-1}_i(P)$ in $Z_i$ in Look phase.
Since $P$ and $Q^{(i)}$ are similar,
$k_P = k_{Q^{(i)}}$, i.e., all robots can consistently compute $k_P$.

We introduce a total order $\succ$ on $\overline{P}$,
which is denoted by $\succ_P$, 
in such a way that all robots can agree on the order,
provided $k_P = 1$.
A key idea behind the definition of $\succ_P$ is to use, instead of $Z_i$,
an $x$-$y$ coordinate system $\Xi_i$ which is computable 
for any robot $r_j$ from $Q^{(j)}$.

Let $\Gamma_P(\bm{q}) \subseteq \overline{P}$ be the orbit of $G_P$ 
through $\bm{q} \in \overline{P}$. 
Then $|\Gamma_P(\bm{q})| = k_P$ if $\bm{q} \not= \bm{o}_P$,
and $\mu_P(\bm{q}') = \mu_P(\bm{q})$ if $\bm{q}' \in \Gamma_P(\bm{q})$.
If $\bm{o}_P \in \overline{P}$, $\Gamma_P(\bm{o}_P) = \{ \bm{o}_P \}$.
Let $\Gamma_P = \{ \Gamma_P(\bm{q}): \bm{q} \in \overline{P} \}$.
Then $\Gamma_P$ is a partition of $\overline{P}$.
Define $x$-$y$ coordinate system $\Xi_{\bm{q}}$ for any point 
$\bm{q} \in \overline{P} \setminus \{ \bm{o}_P \}$.
The origin of $\Xi_{\bm{q}}$ is $\bm{q}$,
the unit distance is the radius of the smallest enclosing circle of $P$,
the $x$-axis is taken so that it goes through $\bm{o}_P$,
and it is right-handed.
Let $\gamma_{\bm{q}}$ be the coordinate transformation from $\Xi_{\bm{q}}$ to $Z_0$.
Then the view $V_P(\bm{q})$ of $\bm{q}$ is defined to be $\gamma^{-1}_{\bm{q}}(P)$.
Obviously $V_P(\bm{q}') = V_P(\bm{q})$ (as multisets), 
if and only if $\bm{q}' \in \Gamma_P({\bm{q}})$.
Let $View_P = \{ V_P({\bm{q}}): \bm{q} \in \overline{P} \setminus \{ \bm{o}_P \} \}$.

Any robot $r_i$, in Compute phase from $Q^{(i)}$, 
can compute $\Xi_{\bm{q}}$ and $V_{Q^{(i)}}(\bm{q})$
for each $\bm{q} \in \overline{Q^{(i)}} \setminus \{ \bm{o}_{Q^{(i)}} \}$,
and thus $View_{Q^{(i)}}$.
Since $P$ and $Q^{(i)}$ are similar,
by the definition of $\Xi_{\bm{q}}$, $View_{P} = View_{Q^{(i)}}$,
which implies that all robots $r_i$ can consistently compute $View_{P}$.

We define $\succ_P$ on $\Gamma_P$ using $View_P$.
For any distinct orbits $\Gamma_P(\bm{q})$ and $\Gamma_P(\bm{q}')$,
$\Gamma_P(\bm{q}) \succ_P \Gamma_P(\bm{q}')$, if and only if
one of the following conditions hold:
\begin{enumerate}
\item 
$\mu_P(\bm{q}) > \mu_P(\bm{q}')$.

\item
$\mu_P(\bm{q}) = \mu_P(\bm{q}')$ 
and $dist(\bm{q}, \bm{o}_P) < dist(\bm{q}', \bm{o}_P)$ hold,
where $dist(\bm{x}, \bm{y})$ is the Euclidean distance 
between $\bm{x}$ and $\bm{y}$.

\item
$\mu_P(\bm{q}) = \mu_P(\bm{q}')$,
$dist(\bm{q}, \bm{o}_P) = dist(\bm{q}', \bm{o}_P)$, and
$V_P(\bm{q}) \sqsupset V_P(\bm{q}')$ hold.
\footnote{Since $dist(\bm{o}_P, \bm{o}_P) = 0$,
$V_P({\bm{q}})$ is not compared with $V_P(\bm{o}_P)$ with respect to $\sqsupset$.}
\end{enumerate}
Then $\succ_P$ is a total order on $\Gamma_P$.
If $k_P = 1$, since $\Gamma_P(\bm{q}) = \{ \bm{q} \}$ 
for all $\bm{q} \in \overline{P}$,
we regard $\succ_P$ as a total order on $\overline{P}$
by identifying $\Gamma_P({\bm{q}})$ with $\bm{q}$.
For a configuration $P$ (in $Z_0$), from $Q^{(i)}$ (in $Z_i$),
each robot $r_i$ can consistently compute $k_P = k_{Q^{(i)}}$,
$\Gamma_P = \Gamma_{Q^{(i)}}$, and $View_P = View_{Q^{(i)}}$, 
and hence $\succ_P = \succ_{Q^{(i)}}$.
Thus, all robots can agree on, e.g., 
the largest point $\bm{q} \in \overline{P}$ with respect to $\succ_P$.

Since $k_P = 0$ implies $m_P = 1$, i.e., the FC(2)-PO has been solved,
we may assume $k_P \geq 1$ in what follows.
We partition the set of all multisets $P = \{ \bm{p}_1, \bm{p}_2, \ldots , \bm{p}_n \}$ 
for all $n \geq 4$ into six types G, L, T, I, S, and Z.
Let $m_P = |\overline{P}|$.

\begin{description}
\item[G(oal):]
$m_P \leq 2$.

\item[L(ine):]
$CH(P)$ is a line segment.

\item[T(riangle):]
$m_P = 3$ and $CH(P)$ is a triangle.

\item[I(nside):]
$m_P = 4$, $CH(P)$ is a triangle,
and $\bm{o}_P \in P$.

\item[S(ide):]
$m_P = 4$, $CH(P)$ is a triangle,
and $\bm{M}_P \in P$,
where $\bm{M}_P$ is the middle point of a longest side of $CH(P)$.

\item[Z:]
$P$ does not belong to the above five types.
\end{description}

Now we define target function $\psi_{(n,2)}$.
Let $\mathcal P$ be the set of multisets that contains at least one $(0,0)$, 
which is the domain of a target function $\psi_{(n,2)}$.
Algorithm LN$_{(n,2)}$, which we assume to exist, transforms any configuration
of type L into a configuration of type G, 
without going through a configuration whose type is not L.

\medskip
\noindent
{\bf [Target function $\psi_{(n,2)}$]}
\begin{enumerate}
\item
When $P$ is type Z:
\begin{description}
\item[(a)] 
If $k_P \geq 2$, $\psi_{(n,2)}(P) = \bm{o}_P$.
\item[(b)]
If $k_P = 1$, $\psi_{(n,2)}(P) = \bm{a}_P$,
where $\bm{a}_P \in \overline{P}$ is the largest point with respect to $\succ_P$,
which is well-defined since $k_P = 1$.
\end{description}

\item
If $P$ is type L, invoke LN$(n,2)$.

\item
When $P$ is type T,
let $\overline{P} = \{ \bm{a}, \bm{b}, \bm{c} \}$.
\begin{description}
\item[(a)]
If triangle $\bm{abc}$ is equilateral, $\psi_{(n,2)}(P) = \bm{o}_P$.

\item[(b)]
If triangle $\bm{abc}$ is not equilateral, $\psi_{(n,2)}(P) = \bm{M}_P$,
where $\bm{M}_P$ is the middle point of the longest side.
If there are two longest sides,
$\bm{M}_P$ is the middle point of the side next to the shortest side counter-clockwise.
\end{description}

\item
If $P$ is type I, $\psi_{(n,2)}(P) = \bm{o}_P$.

\item
If $P$ is type S, $\psi_{(n,2)}(P) = \bm{M}_P$
(which is defined in the definition of type S).

\end{enumerate}

To show that target function $\psi_{(n,2)}$ is an algorithm 
to solve the  FC(2)-PO for $n \geq 4$, 
provided the existence of LN$_{(n,2)}$,
we need the following technical lemma.
A multiset $P$ is said to be {\em linear}
if $CH(P)$ is a line segment.

\begin{lemma}
\label{A1005} 
Let $A$ be a {\bf set} (not a multiset) of points satisfying 
(1) $A$ is not linear,
(2) $k_A \geq 2$, 
and (3) $\bm{o}_A \not\in A$.
For any point $\bm{a} \in A$,
let $B = (A \cup \{ \bm{o}_A \}) \setminus \{ \bm{a} \}$,
i.e., $B$ is constructed from $A$ by replacing $\bm{a} \in A$ with $\bm{o}_A$.
Then $k_B = 1$.
\end{lemma}

\begin{proof}
For any finite set $S$,
$m_S = |S|$ is the size of $S$,
$C_S$ is the smallest enclosing circle of $S$,
whose center is $\bm{o}_S$ and radius is $d_S$,
$k_S$ is the order of the rotation group $G_S$ of $S$ around $\bm{o}_S$,
and $CH_S$ is the convex hull of $S$.
We introduce a notation for the convenience of description.
Let $C$ be a circle, 
and $\bm{x}$ and $\bm{y}$ be points on $C$.
By $C(\bm{x},\bm{y})$ (resp. $C[\bm{x},\bm{y}]$),
we denote the arc of $C$ from $\bm{x}$ to $\bm{y}$ counter-clockwise,
excluding (resp. including) both ends $\bm{x}$ and $\bm{y}$.

Consider any finite set $A$ which satisfies the following conditions:
\begin{enumerate}
\item
$A$ is not linear, i.e., $CH_A$ is not a line segment,
\item
$k_A \geq 2$, and
\item
$\bm{o}_A \not\in A$.
\end{enumerate}

Since $k_A \geq 2$ and $A$ is not linear, $m_A \geq 3$.
For any point $\bm{a} \in A$,
let $B = (A \cup \{ \bm{o}_A \}) \setminus \{ \bm{a} \}$.
We show that $k_B = 1$.
Proof is by contradiction.
We assume $k_B \geq 2$ to derive a contradiction.

\medskip
\noindent
(I) We first show that $\bm{o}_A = \bm{o}_B$.
Proof is by contradiction.
We assume $\bm{o}_B \not= \bm{o}_A$ to derive a contradiction.

If $k_A \geq 4$, then $C_A = C_B$ and hence 
$\bm{o}_A = \bm{o}_B$.
Thus, $2 \leq k_A \leq 3$, 
and the number of points in $A$ on $C_A$ 
($\bm{a}$ must be one of them) is at most 3.
Without loss of generality,
we assume that $\bm{o}_A = (0,0)$ and $\bm{o}_B = (-1,0)$.
Then the $x$-coordinate of $\bm{a}$ is positive,
since $d_B \leq d_A$.

\medskip
\noindent
(A) First consider the case in which $k_A = 2$.
We start with showing that $(2 \leq)$ $k_B \leq 3$.
There are exactly two points $\bm{a}$ and $-\bm{a}$ on $C_A$,
where $-\bm{a}$ is the opposite point of $\bm{a}$ about $\bm{o}_A$.
Thus $-\bm{a}$ is on $C_A$ and in $A$.
Let $\bm{c} = (c_x, c_y)$ and $\bm{c}' = (c_x, -c_y)$
be the intersections of $C_A$ and $C_B$,
where $c_y \geq 0$ ($\bm{c} = \bm{c}'$, i.e., $c_y = 0$ may occur).
Since no points in $B$ are on $C_B(\bm{c}, \bm{c}')$,
if $c_x >  -1$ held, $C_B$ would not be the smallest enclosing circle.
Thus $c_x \leq -1$ holds.

We examine where a point in $B$ occurs on $C_B$.
Let $\bm{h}$ and $-\bm{h}$ be the intersections of $C_B$ and the $y$-axis,
i.e., $\bm{h} = (0, \sqrt{d^2_B - 1})$ and $-\bm{h} = (0, -\sqrt{d^2_B - 1})$.
Since $\bm{o}_A \in B$, $d_B \geq 1$, 
and there indeed exist
$\bm{h}$ and $-\bm{h}$ (which may be the same).
First, there is a point in $B$ on $C_B[\bm{h},-\bm{h}]$,
since otherwise, $C_B$ is not the smallest enclosing circle of $B$.
It cannot occur on $C_B(\bm{c}, \bm{c}')$, 
since $C_B(\bm{c}, \bm{c}')$ is located outside $C_A$.
Furthermore, 
it does not occur either on $C_B(\bm{h}, \bm{c})$ or on $C_B(\bm{c}', -\bm{h})$.
If a point $\bm{x} \in B$ occurred there,
$d_B = dist(\bm{x},\bm{o}_B) < dist(-\bm{x},\bm{o}_B)$
and $-\bm{x} \in B$ would hold.
Note that $\bm{x}$ and thus $-\bm{x}$ is not on $C_A$,
and hence $-\bm{x}$ is not $\bm{a}$.
Thus if $\bm{x} \in B$ is a point on $C_B[\bm{h},-\bm{h}]$,
it is either $\bm{c}$ or $\bm{c}'$.

If both of $\bm{c}$ and $\bm{c}'$ were in $B$,
since they were also on $C_A$, $k_A = 2$, and $c_x \leq -1$,
$\bm{o}_A = \bm{o}_B$ would hold.
Thus exactly one point in $B$ is on $C_B[\bm{h},-\bm{h}]$,
which immediately implies $k_B \leq 3$.

\medskip
\noindent
(A1) Consider the case in which $k_A = k_B = 2$.
Suppose that $\bm{a} \not= (d_A, 0)$.
For any point $\bm{p}_0 \in B$,
define a sequence of points
\[
{\cal X}: \bm{p}_0, \bm{q}_0, \bm{p}_1, \bm{q}_1, \ldots
\] 
as follows:
For any $i \geq 0$, 
$\bm{q}_i$ is the opposite point of $\bm{p}_i$ about $\bm{o}_B$,
and for any $ i \geq 1$, 
$\bm{p}_i$ is the opposite point of $\bm{q}_{i-1}$ about $\bm{o}_A$.
If $\bm{p}_i \not= \bm{a}$ and $\bm{q}_i \not= \bm{o}_A$ for all $i \geq 0$,
then $\bm{p}_i, \bm{q}_i \in A \cap B$,
since $k_A = k_B = 2$.

Consider an instance of ${\cal X}$ for $\bm{p}_0 = \bm{o}_A$.
First $\bm{q}_i \not= \bm{o}_A$ by definition.
Next $\bm{p}_i \not= \bm{a}$,
since all $\bm{p}_i$ and $\bm{q}_i$ occur on the $x$-axis,
and $\bm{a}$ is not on the $x$-axis.
Thus ${\cal X}$ consists of an infinite number of distinct points.
It is a contradiction, 
since there are only $m_A$ points in $A$ (and $B$).

Suppose otherwise that $\bm{a} = (d_A,0)$.
There is a point $\bm{p} = (p_x, p_y) \in A$
such that $p_x \geq 0$, $p_y \not= 0$, and it is not on $C_A$, 
since $A$ is not linear and $k_A = 2$.
Since $p_y \not= 0$, $\bm{p}$ is neither $\bm{o}_A$ nor $\bm{a}$,
and thus $\bm{p} \in B$, as well. 
Consider another instance of ${\cal X}$ for $\bm{p}_0 = \bm{p}$.
Let $\bm{p}_i = (p^i_x,p^i_y)$ and $\bm{q}_i = (q^i_x,q^i_y)$ for $i \geq 0$.
By a simple induction, for all $i \geq 0$,
$p^{i+1}_x > p^i_x \geq 0$, $p^i_y \neq 0$, $q^{i+1}_x < q^i_x < -1$, $q^i_y \neq 0$, and
$\bm{p}_i, \bm{q}_i \in A \cap B$.
It is a contradiction, 
since $\bm{p}_i \not= \bm{a}$ and $\bm{q}_i \neq \bm{o}_A$ for all $i\geq 0$;
${\cal X}$ consists of an infinite number of distinct points.

\medskip
\noindent
(A2) Consider the case in which $k_A = 2$ and $k_B = 3$.
We again consider the sequence 
${\cal X}: \bm{p}_0, \bm{q}_0, \bm{p}_1, \bm{q}_1, \ldots$
after making two modifications.
Since $k_B = 3$,
we have two candidates to determine $\bm{q}_i$ from $\bm{p}_i$.
Let $C$ be the circle with center $\bm{o}_B$ 
such that it contains $\bm{p}_i$.
Let $\bm{x}$ and $\bm{x}'$ be two points on $C$ such that 
they form an equilateral triangle with $\bm{p}_i$.
Since $k_B = 3$, they both belong to $B$.
We assume that $\bm{x}$ has a smaller $x$-coordinate than $\bm{x}'$
(in case of a tie, we assume that 
$\bm{x}$ has a smaller $y$-coordinate than $\bm{x}'$).
Then we choose $\bm{x}$ as $\bm{q}_i$.

We consider two instances of ${\cal X}$.
${\cal X}_1$ starts with $\bm{p}_0 = \bm{o}_A$ and $\bm{q}_0 = (-3/2, \sqrt{3}/2)$,
and ${\cal X}_2 $ with $\bm{p}_0 = \bm{o}_A$ and $\bm{q}_0 = (-3/2, -\sqrt{3}/2)$.
Then, $\bm{p}_i = (p^i_x,p^i_y)$ (resp. $\bm{q}_i = (q^i_x,q^i_y)$) in ${\cal X}_1$, 
if and only if $\bm{p}_i = (p^i_x,-p^i_y)$ (resp. $\bm{q}_i = (q^i_x,-q^i_y))$ in ${\cal X}_2$.

Consider two straight lines $\ell_1: y = \sqrt{3} (x+1)$ and $\ell_2: y = -\sqrt{3}(x+1)$.
By a similar induction like (A1), 
we can show that for all $i \geq 1$,
$p^{i+1}_x > p^i_x > 0$, $p^i_y \not= 0$, $q^{i+1}_x < q^i_x < -1$, $q^i_y \not= 0$,
since $\bm{p}_i$ (resp. $\bm{q}_i$) is located to the right (resp. left) side 
of $\ell_1$ and $\ell_2$.
Since $\bm{q}_i$ do not reach $\bm{o}_A$,
${\cal X}_1$ and ${\cal X}_2$ both create an infinite number of distinct points,
if they do not reach $\bm{a}$.

By definition,
$\bm{p}_i = (p^i_x,p^i_y)$ (resp. $\bm{q}_i = (q^i_x,q^i_y)$) in ${\cal X}_1$,
if and only if $\bm{p}_i = (p^i_x,-p^i_y)$ (resp. $\bm{q}_i = (q^i_x,-q^i_y)$) in ${\cal X}_2$,
and $p^i_y, q^i_y \not= 0$, for all $i \geq 1$.
Thus at least one of them creates an infinite number of distinct points,
no matter where $\bm{a}$ is.
It is a contradiction.

\medskip
\noindent
(B) Next consider the case in which $k_A = 3$.
First we show that $k_B \leq 2$.
Let $\bm{p}, \bm{x}$, and $\bm{x}'$ be the three points on $C_A$,
which form an equilateral triangle.
Recall that $\bm{c}$ and $\bm{c}'$ are intersections of $C_A$ and $C_B$.
By a similar argument in (A),
one of $\bm{x}$ or $\bm{x}'$ is either $\bm{c}$ or $\bm{c}'$.
Without loss of generality,
we assume that $\bm{x} = \bm{c}$.
Since $\bm{x}' \in B$, 
it occurs inside $C_B$ (and on $C_A$).
Thus it occurs on $C_A[\bm{c}, \bm{c}']$.
Since $\angle \bm{x} \bm{o}_A \bm{x}' = 2\pi/3$,
$\angle \bm{c} \bm{o}_B \bm{c}' > 2\pi/3$,
and there is no point in $B$ on $C_B[\bm{c}, \bm{c}']$,
$k_B \leq 2$ holds.
All what we need to consider is thus the case in which $k_A = 3$ and $k_B = 2$.

We again consider the sequence 
${\cal X}: \bm{p}_0, \bm{q}_0, \bm{p}_1, \bm{q}_1, \ldots$
after making two modifications.
Since $k_A = 3$ this time,
we have two candidates to determine $\bm{p}_i$ from $\bm{q}_{i-1}$.
Let $C$ be the circle with center $\bm{o}_A$ 
such that it contains $\bm{q}_{i-1}$.
Let $\bm{x}$ and $\bm{x}'$ be two points on $C$ such that 
they form an equilateral triangle with $\bm{q}_{i-1}$.
Since $k_A = 3$, they both belong to $A$.
We assume that $\bm{x}$ has a larger $x$-coordinate than $\bm{x}'$
(in case of tie, we assume that $\bm{x}$ has a larger $y$-coordinate than $\bm{x}'$).
Then we choose $\bm{x}$ as $\bm{p}_i$.

We consider two instances of ${\cal X}$.
${\cal X}_1$ starts with $\bm{p}_0 = \bm{o}_A$, $\bm{q}_0 = (-2,0)$,
and $\bm{p}_1 = (1,\sqrt{3})$,
and ${\cal X}_2 $ with $\bm{p}_0 = \bm{o}_A$, $\bm{q}_0 = (-2,0)$,
and $\bm{p}_1 = (1, -\sqrt{3})$.
Then  $\bm{p}_i = (p^i_x,p^i_y)$ (resp. $\bm{q}_i = (q^i_x,q^i_y)$) in ${\cal X}_1$,
if and only if $\bm{p}_i = (p^i_x,-p^i_y)$ (resp. $\bm{q}_i = (q^i_x,-q^i_y)$) in ${\cal X}_2$,
and $p^i_y, q^i_y \not= 0$ for all $i \geq 1$.
By a similar argument in (A2),
they both create an infinite number of distinct points,
if they do not reach $\bm{a}$.
Thus at least one of them creates an infinite number of distinct points,
no matter where $\bm{a}$ is.
It is a contradiction.

\medskip
\noindent
(II) Thus, $\bm{o}_A = \bm{o}_B$, if $k_B \geq 2$.
It is however a contradiction.
Without loss of generality, we may assume that $\bm{o}_A = \bm{o}_B = (0,0)$.
Since $k_A, k_B \geq 2$,
$\sum_{\bm{x} \in A} \bm{x} = \sum_{\bm{y} \in B} \bm{y} = (0,0)$.
It is a contradiction, 
since $\sum_{\bm{y} \in B} \bm{y} =  (\sum_{\bm{x} \in A} \bm{x}) - \bm{a} + (0,0)$
and $\bm{a} \not= \bm{o}_A = (0,0)$.
\end{proof}

\begin{lemma}
\label{L4040}
Target function $\psi_{(n,2)}$ is an algorithm to solve the \mbox{\rm FC(2)-PO} for $n \geq 4$,
provided the existence of LN$_{(n,2)}$.
\end{lemma}

\begin{proof}
Suppose that a robot $r_i$ starts a Look-Compute-Move phase 
when the configuration is $P = \{ \bm{p}_1, \bm{p}_2, \ldots, \bm{p}_n \}$ (in $Z_0$).
Let $Z_i$ and $\gamma_i$ be the 
$x$-$y$ local coordinate system of $r_i$,
and the coordinate transformation from $Z_i$ to $Z_0$, respectively.
Then $r_i$ observes 
$Q^{(i)} = \{ \bm{q}_1, \bm{q}_2, \ldots , \bm{q}_n \}$ 
in $Z_i$ in Look phase,
where $\bm{p}_j = \gamma_i(\bm{q}_j)$.
By the definition of $Z_i$, $(0,0) \in Q^{(i)}$, 
and there is an $j$ such that $\gamma_i((0,0)) = \bm{p}_j$. 
($P$ may not contain $(0,0)$.)

In Compute phase, $r_i$ computes $\psi_{(n,2)}(Q^{(i)})$ (not $\psi_{(n,2)}(P)$),
which is the target point of $r_i$ in $Z_i$,
where $\psi_{(n,2)}(Q^{(i)})$ is either $\bm{o}_{Q^{(i)}}$, 
$\bm{M}_{Q^{(i)}}$, or $\bm{a}_{Q^{(i)}}$.
Immediately, $k_{Q^{(i)}}$, $\bm{o}_{Q^{(i)}}$, and $\bm{M}_{Q^{(i)}}$ 
are computable from $Q^{(i)}$,
and $k_{Q^{(i)}} = k_P$, $\gamma_i(\bm{o}_{Q^{(i)}}) = \bm{o}_P$, 
and $\gamma_i(\bm{M}_{Q^{(i)}}) = \bm{M}_P$.

To compute $\bm{a}_{Q^{(i)}}$ for Step~1(b),
$r_i$ computes $View_{Q^{(i)}}$.
Indeed, it is possible, 
since $r_i$ can construct $\Xi_{\bm{q}}$ for all 
$\bm{q} \in  Q^{(i)} \setminus \{ \bm{o}_{Q^{(i)}} \}$.
Then it computes $\succ_{Q^{(i)}}$ and $\bm{a}_{Q^{(i)}}$.
By the definition of $\Xi_{\bm{q}}$, $View_P = View_{Q^{(i)}}$,
which implies $\succ_{Q^{(i)}} = \succ_P$ and $\gamma_i(\bm{a}_{Q^{(i)}}) = \bm{a}_P$.
Hence when $\psi_{(n,2)}(Q^{(i)})$ is $\bm{o}_{Q^{(i)}}, \bm{M}_{Q^{(i)}}$, 
or $\bm{a}_{Q^{(i)}}$ in $Z_i$, 
$r_i$ moves to $\bm{o}_P$, $\bm{M}_P$, or $\bm{a}_P$ in $Z_0$, respectively.
Using this relation, we analyze $\psi_{(n,2)}$.

Consider any execution ${\cal E}: P_0, P_1, \ldots$ of $\psi_{(n,2)}$,
starting from any initial configuration $P_0$ (in $Z_0$).
We show that in $\cal E$,
each robot converges to one of at most two convergence points,
provided that at most two robots crash.
It is sufficient to show that $\cal E$ eventually reaches a configuration $P_t$ of type L,
since LN$_{(n,2)}$ is invoked at $t$ and solves FC(2)-PO,
without reaching a configuration not in type L.

Our proof scenario is as follows:
Using an exhaustive search,
we figure out how types of configurations change as the execution evolves.
The crucial observations for $\psi_{(n,2)}$,
which we shall make, are 
(i) the transition diagram among the types has a single sink L,
(ii) every loop eventually terminates if the execution does not converge to a point or two, and
(iii) some transition eventually occurs from each type.
Based on these, 
we conclude that the execution eventually reaches a type L configuration, 
if it does not converge to one or two convergence points.

The search is exhaustive and is not very difficult,
if we gaze the effects of the fairness of scheduler and faulty robots.
Consider any configuration $P_t$ not in type L.
If every robot $r$ activated at $t$ is either faulty
or $\psi_{(n,2)}$ instructs it not to move,
$P_{t+1} = P_t$ holds.

A robot $r$ is said to be \emph{ready} at $P_t$, 
if it moves once it is activated.
We first confirm that at least three ready robots exist,
in the following observations.
(This confirmation is easy, and we omit to explicitly mention it.)
Then at least one non-faulty ready robot $r$ exists in $P_t$.
Since the scheduler is fair, $r$ is activated eventually at some time $t' > t$.
If $P_{t'} = P_t$ and $r$ is activated at $t'$, at least $r$ moves,
and the execution reaches a configuration $P_{t'+1}$.
(Note that $P_{t'+1} = P_{t'}$ may hold, 
although some robots move, e.g., when two robots exchange their positions.)
Thus, if there are three ready robots in $P_t$,
in every execution starting from $P_t$,
there is a time $t' > t$ such that there are robots that move at $t'$.
We therefore assume that some robots always move every time $t$ 
without loss of generality,
provided that there are at least three ready robots.

We first make a series of observations (I)-(VII).
Figure~\ref{F0020} shows the transition diagram among the types,
which summarizes the observations (I)--(V).
Let $k_t = k_{P_t}, m_t = m_{P_t}$, $CH_t = CH(P_t)$, and
$\bm{o}_t$ be the center of the smallest enclosing circle $C_t$ of $P_t$.

\begin{figure}[t]
\centering
\includegraphics[scale=1]{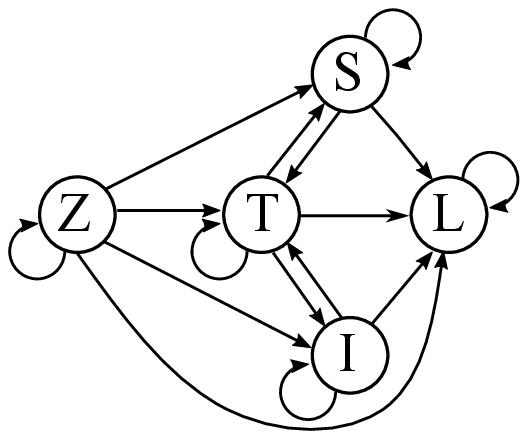}
\caption{The transition diagram among types that $\psi_{(n,2)}$ specifies.}
\label{F0020}
\end{figure}

\begin{description}
\item[(I) When $P_t$ is type Z.]

Suppose that $P_t$ is type Z.
Then $P_{t+1}$ can be any type in L, T, I, S, and Z.
(We regard a type G configuration as type L.)
Suppose that $P_{t+1}$ is type Z.
We show that this self-loop from Z to Z eventually terminates.

\medskip
\noindent
(A) Suppose that $k_t = 1$.
Let $\bm{a}_t$ be the largest point in $\overline{P_t}$ with respect to $\succ_{P_t}$,
and let $\gamma_i$ be the coordinate transformation from $Z_i$ to $Z_0$.
When a robot $r_i$ is activated, in Look phase, it identifies 
a configuration $Q^{(i)}_t$ such that $\gamma_i(Q^{(i)}_t) = P_t$,
and computes $\psi_{(n,2)}(Q^{(i)}_t) = \bm{a}^{(i)}_t$ in Compute phase.
Since $Q^{(i)}_t$ is type Z, and $k_{Q^{(i)}_t} = 1$,
$\bm{a}^{(i)}_t$ is the largest point in $\overline{Q^{(i)}_t}$ with respect to 
$\succ_{Q^{(i)}_t} (= ~\succ_{P_t})$.
Thus $\gamma_i(\bm{a}^{(i)}_t) = \bm{a}_t$.

By the definition of $\succ_{P_t}$,
$\mu_{P_t}(\bm{a}_t) \geq \mu_{P_t}(\bm{p})$ for all $\bm{p} \in \overline{P_t}$.
Thus, $\mu_{P_{t+1}}(\bm{a}_t) > \mu_{P_t}(\bm{a}_t)$
$\geq \mu_{P_t}(\bm{p}) \geq \mu_{P_{t+1}}(\bm{p})$, 
for all $\bm{p} \in \overline{P_t} \setminus \{ \bm{a}_t \}$.
Therefore, if $P_{t+1}$ is type Z,
$k_{t+1} = 1$ and $\bm{a}_{t+1} = \bm{a}_t$.

Suppose that $P_{t'}$ is type Z with $k_{t'} = 1$ for all $t' > t$.
Then a contradiction is derived,
since $\bm{a}_{t'} = \bm{a}_t$ and $\mu_{P_{t'}}(\bm{a}_t)$ increases unboundedly.
Thus the self-loop of Z eventually terminates.

\medskip
\noindent
(B) Suppose that $k_t \geq 2$.
We show that, if $P_{t'}$ is type Z for all $t' > t$,
then there is a time $t'$ such that $k_{t'} = 1$ or $m_{t'} < m_t$ holds,
which implies that the self-loop of Z eventually terminates by (A).

Suppose that $k_{t+1} \geq 2$.
Then $m_{t+1} \leq m_t + 1$.
If $m_{t+1} < m_t$, there is nothing to show.
Thus there are two cases to be considered.

First consider the case of $m_{t+1} = m_t$.
If $\overline{P_{t+1}} = \overline{P_t}$,
then $\bm{o}_{t+1} = \bm{o}_t$,
and hence $\mu_{t+1}(\bm{o}_{t+1}) = \mu_{t+1}(\bm{o}_t) > \mu_t(\bm{o}_t)$.
Thus, there is a time $t' > t+1$ such that $m_{t'} < m_{t+1} = m_t$.

Otherwise if $\overline{P_{t+1}} \not= \overline{P_t}$,
$\bm{o}_t \not\in \overline{P_t}$, 
and there is a $\bm{q} \in \overline{P_t}$
such that $\overline{P_{t+1}} = (\overline{P_t} \setminus \{ \bm{q} \}) \cup \{ \bm{o}_t \}$.
By Lemma~\ref{A1005}, $k_{t+1} = 1$.

Next consider the case of $m_{t+1} = m_t + 1$.
Then $\bm{o}_t \not\in \overline{P_t}$, 
and $\overline{P_{t+1}} = \overline{P_t} \cup \{ \bm{o}_t \}$.
Since $\bm{o}_{t+1} = \bm{o}_t$ and 
$\mu_{t+1}(\bm{o}_{t+1}) = \mu_{t+1}(\bm{o}_t) > \mu_t(\bm{o}_t)$,
there is a time $t' > t+1$ such that $m_{t'} < m_{t+1}$.
Without loss of generality, 
we assume $m_{t+1} = m_{t+2} = \dots = m_{t'-1}$,
i.e., $\overline{P_{t+1}} = \overline{P_{t+2}} = \dots = \overline{P_{t'-1}}$
and hence $\bm{o}_{t+1} = \bm{o}_{t+2} = \dots = \bm{o}_{t'-1}$.
If $m_{t'} < m_t$, there is nothing to show.

Suppose that $m_{t'} = m_t$.
By assumption, at each time $t'' = t, t+1, \ldots , t'-1$,
some robots not at $\bm{o}_t$ move to $\bm{o}_t$.
Let $A$ be the set of robots activated in this period.
If all robots in $A$ are activated at $t$, $P_{t'}$ yields.
By Lemma~\ref{A1005}, $k_{t'} = 1$.

Thus the self-loop of Z eventually terminates.

\item[(II) When $P_t$ is type L.]

Then LN$_{(n,2)}$ is invoked to solve FC(2)-PO.
Recall that $P_t$ is type L, so is $P_{t+1}$.
Thus $L$ is a sink in the transition diagram.

\item[(III) When $P_t$ is type T.]
Let $\overline{P_t} = \{ \bm{a}, \bm{b}, \bm{c} \}$.

\medskip
\noindent
(A) Suppose that triangle $\bm{abc}$ is equilateral.
Since $\psi_{(n,2)}(P_t) = \bm{o}_t$,\footnote{
Formally, the equality ``$\psi_{(n,2)}(P_t) = \bm{o}_t$'' does not hold,
since $P_t$ may not contain $(0,0)$,
despite that the domain of target function $\psi_{(n,2)}$ is $\mathcal P$.
Here and later when notation $\psi_{(n,2)}(P_t)$ appears,
recall the discussion at the beginning of the proof.
On robot $r_i$, $\psi_{(n,2)}$ is not applied to $P_t$, but to $Q^{(i)}_t$,
where $\gamma_i(Q^{(i)}_t) = P_t$.
In this case, $\psi_{(n,2)}(Q^{(i)}_t) = \bm{o}_{Q^{(i)}_t}$.
Since $\gamma_i(\bm{o}_{Q^{(i)}_t}) = \bm{o}_t$,
the target point of $r_i$ is $\bm{o}_t$ in $Z_0$.
Formally, $\gamma_i(\psi_{(n,2)}(\gamma^{-1}_i(P_t))) = \bm{o}_t$.
This is what ``$\psi_{(n,2)}(P_t) = \bm{o}_t$'' means.
The same convention applies to LN$_{(n,2)}$ later in Section 7.2.3.}
every non-faulty robot, once activated, moves to $\bm{o}_t$.
Thus the type of $P_{t+1}$ is L, T, or I.
If the type of $P_{t+1}$ is T,
then $CH_{t+1}$ is not equilateral,
and $h_{t+1} = \frac{2\sqrt{3}+3}{9} h_t$,
where $h_t = h(CH_t)$ is the perimeter of $CH_t$.

\medskip
\noindent
(B) Suppose that triangle $\bm{abc}$ is not equilateral.
Since $\psi_{(n,2)}(P_t) = \bm{M}_t$,
and all $Z_i$ are right-handed,
the type of $P_{t+1}$ is L, T, or S.
where $\bm{M}_t = \bm{M}_{P_t}$.
If the type of $P_{t+1}$ is T, then $h_{t+1} < h_t$.

\medskip
We show that $CH_t$ (and thus $P_t$) converges to a point,
if the self-loop of T does not terminate.
Suppose that $P_t$ is type T for all $t > t_0$ for some time $t_0$.
After time $t_0$, $CH_t$ is equilateral at a finite number of times,
since if $CH_t$ is equilateral then $h_{t+1} = \frac{2\sqrt{3}+3}{9} h_t$.
Thus, there is a time $t_1$ such that $CH_t$ is not equilateral for all $t > t_1$.
Since $h_{t+2} \leq \frac{5}{6} h_t$ by a simple calculation,
we conclude that $P_t$ converges to a point.

\item[(IV) When $P_t$ is type I.]

Since $\psi_{(n,2)}(P_t) = \bm{o}_t$,
the type of $P_{t+1}$ is either L, T, or I.

Suppose that $P_{t+1}$ is type I.
We show that this self-loop from I to I eventually terminates.

If $P_{t+1}$ is type I,
$\bm{o}_{t+1} = \bm{o}_t$,
and hence $\mu_{P_{t+1}}(\bm{o}_{t+1}) = \mu_{P_{t+1}}(\bm{o}_t) > \mu_{P_t}(\bm{o}_t)$.
Thus this self-loop of I eventually terminates
(and the execution will reach a type L or T configuration).

If $P_{t+1}$ is type T,
then like (III),
$h_{t+1} = \frac{2\sqrt{3}+3}{9} h_t$.

\item[(V) When $P_t$ is type S.]

Since $\psi_{(n,2)}(P_t) = \bm{M}_t$,
the type of $P_{t+1}$ is either L, T, or S.

If $P_{t+1}$ is type S,
$\bm{M}_{t+1} = \bm{M}_t$,
and hence $\mu_{P_{t+1}}(\bm{M}_{t+1}) = \mu_{P_{t+1}}(\bm{M}_t) > \mu_{P_t}(\bm{M}_t)$.
Thus this self-loop from S to S eventually terminates
(and the execution will reach a type L or T configuration).

\item[(VI) TI$^+$T-loop.]

Suppose that $P_t$ is type T.
If triangle $\bm{abc}$ is equilateral,
the execution may reach a configuration $P_{t'}$ of type T
via several type I configurations as observed in (III) and (IV).
Now $\overline{P_{t'}} = \{ \bm{x}, \bm{y}, \bm{o}_t \}$,
where $\bm{x}, \bm{y} \in \{ \bm{a}, \bm{b}, \bm{c}\}$ and $\bm{x}\neq \bm{y}$.
Thus 
$h_{t'} = \frac{2\sqrt{3}+3}{9} h_t$.

\item[(VII) TS$^+$T-loop.]

Suppose again that $P_t$ is type T.
If triangle $\bm{abc}$ is not equilateral,
the execution may reach a configuration $P_{t'}$ of type T
via several type S configurations as observed in (III) and (V).
Now $\overline{P_{t'}} = \{ \bm{x}, \bm{y}, \bm{M}_t \}$,
where $\bm{x}, \bm{y} \in \{ \bm{a}, \bm{b}, \bm{c}\}$ and $\bm{x} \neq \bm{y}$.
\end{description}

We go on the proof.
The observations (I)--(V) summarized in Figure~\ref{F0020} show
that any execution $\cal E$ starting from a configuration of any type
eventually reaches a type L configuration,
if neither TI$^+$T-loop nor TS$^+$T-loop repeats infinitely many times.

Let us consider what happens if TI$^+$T-loop or TS$^+$T-loop repeats infinitely many times.
Since $\psi_{(n,2)}(P) \in CH(P)$ for all $P \in {\mathcal P}$,
$h_{t+1} \leq h_t$ for all $t$.
If TI$^+$T-loop occurs infinitely many times in $\cal E$,
$h_t$ converges to 0 by  observation (VI),
and hence $P_t$ converges to a point 
(although $\cal E$ may not contain a configuration of type L or G).

We assume that $\cal E$ contains a finite number of occurrences of TI$^+$T-loop.
Then there is a $t_0$ such that the 
postfix ${\cal E}'$ of $\mathcal{E}$: 
$P_{t_0}, P_{t_0+1}, \ldots$ of $\cal E$ does not contain an occurrence of TI$^+$T-loop.
Since ${\cal E}'$ does not contain a type L configuration,
it is a repetition of TS$^+$T-loop.
Suppose that at time $t$ a TS$^+$T-loop starts and 
at time $t' (> t)$ the second TS$^+$T-loop ends (counting after $t$).
Then it is easy to observe that 
$h_{t'} \leq \frac{5}{6} h_t$ and thus $\cal E$' 
(and hence $\cal E$) converges to a point as observed in (III).

Thus $\cal E$ eventually reaches a configuration of type L,
if it does not converges to a point.
\end{proof}

\subsubsection{Algorithm LN$_{(n,2)}$}
\label{SSSLNn2}

We propose the target function LN$_{(n,2)}$ and show the following:
LN$_{(n,2)}$ is an FC(2)-PO algorithm for any initial configuration 
$P_0$ such that $CH_0$ is a line segment,
where $CH_t = CH(P_t)$.
Moreover, $CH_t \subseteq CH_0$ holds for all $t \geq 0$.
We borrow some symbols and notations from the last section.

Let $P = \{ \bm{p}_1, \bm{p}_2, \ldots , \bm{p}_n \} \in {\mathcal P}$ 
be a configuration of type L, which may be a configuration that a robot 
identifies in Look phase.
We identify a point $\bm{p}_i$ in $R^2$ with a point in $R$:
Since $(0,0) \in P$,
we rotate $P$ about $(0,0)$ counter-clockwise so that the resultant 
$P$ becomes the multiset of points in the $x$-axis.
Then we denote $(p,0)$ by $p$.
In what follows in this section,
a configuration $P$ is thus regarded as a multiset of $n$ real numbers, 
including at least one 0.
We assume $p_1 \leq p_2 \leq \dots \leq p_n$.
By $\overline{P} = \{ b_1, b_2 , \ldots , b_{m_P} \}$, 
we denote the set of distinct real numbers in $P$,
where $m_P$ is the size $|\overline{P}|$ of $\overline{P}$, 
and $b_1 < b_2 < \dots < b_{m_P}$.
The length of $CH(P)$ is denoted by $L_P = b_{m_P} - b_1 = p_n - p_1$.
Let $\lambda_P = \max_{p \in P} \min \{ p - p_1, p_{m_P} - p \} \leq L_P/2$.
Define $j^*$ by $b_{j^*} = 0$.
(Thus the current position of a robot $r_i$ who identifies 
$P$ in Look phase is $b_{j^*}$ in $Z_i$.)
Since $P$ is type L, $k_P \leq 2$.
We denote the middle point of $x$ and $y$ by $M_{xy}$,
i.e., $M_{xy} = (x+y)/2$.

Like $\psi_{(n,2)}$,
we consider ten types,
which we define as follows:

\begin{description}
 \item[G:] 
$m_P \leq 2$. 

\item[B$_3$:]
$m_P = 3$ and $k_P = 2$.

\item[B$_4$:]
$m_P = 4$ and $k_P = 2$.

\item[B$_5$:]
$m_P = 5$ and $k_P = 2$.

\item[B$_6$:]
$m_P = 6$ and $k_P = 2$.

\item[B:]
$m_P \geq 7$ and $k_P = 2$.

\item[U$_3$:]
$m_P = 3$ and $k_P = 1$.

\item[W:]
$m_P = 4$, $k_P = 1$, and
$\overline{P} = \{b_1, b_2, b_3, b_4\} (b_1 < b_2 < b_3 < b_4)$ satisfies
either (a) $2(b_2 - b_1) = b_3 - b_2$ and $b_3 \leq M_{b_1b_4}$, or
(b) $2(b_4 - b_3) = b_3 - b_2$ and $b_2 \geq M_{b_1b_4}$.

\item[U$_4$:]
$m_P = 4$, $k_P = 1$, and
$P$ is not type W.

\item[U:]
$m_P \geq 5$ and $k_P = 1$.
\end{description}

We now give target function LN$_{(n,2)}$.

\medskip
\noindent
{\bf [Target function LN$_{(n,2)}$]}
\begin{enumerate}
\item
If $P$ is type G,
LN$_{(n,2)}(P) = 0$.

\item
When $P$ is type B:
If $j^* \leq \lceil m_P/2 \rceil$, LN$_{(n,2)}(P) = b_1$.
Otherwise if $j^* > \lceil m_P/2 \rceil$, LN$_{(n,2)}(P) = b_{m_P}$.

\item
When $P$ is type B$_3$:
$m_P = 3$.
If $j^* \leq 2$, LN$_{(n,2)}(P) = M_{b_1b_2}$.
Otherwise if $j^* = 3$, LN$_{(n,2)}(P) = M_{b_2b_3}$.

\item
When $P$ is type B$_4$:
$m_P = 4$.
If $j^* \leq 2$, LN$_{(n,2)}(P) = M_{b_1b_2}$.
Otherwise if $j^* \geq 3$, LN$_{(n,2)}(P) = M_{b_3b_4}$.

\item
When $P$ is type B$_5$:
$m_P = 5$.
If $j^* \leq 3$, LN$_{(n,2)}(P) = b_2$.
Otherwise if $j^* \geq 4$, LN$_{(n,2)}(P) = b_4$.

\item
When $P$ is type B$_6$:
$m_P = 6$.
If $j^* \leq 3$, LN$_{(n,2)}(P) = b_2$.
Otherwise if $j^* \geq 4$, LN$_{(n,2)}(P) = b_5$.

\item
When $P$ is type U:
Since $k_P = 1$,
either  $b_1 \succ_P b_{m_P}$ or $b_{m_P} \succ_P b_1$ holds.
If $b_1 \succ_P b_{m_P}$, then LN$_{(n,2)}(P) = b_1$.
Otherwise if $b_{m_P} \succ_P b_1$, LN$_{(n,2)}(P) = b_{m_P}$.

\item
When $P$ is type U$_3$:
Since $k_P = 1$ and $m_P = 3$,
if $b_2 = M_{b_1b_3}$,
then $\mu_P(b_1) \not= \mu_P(b_3)$.
If $b_2 < M_{b_1b_3}$ or
$(b_2 = M_{b_1b_3}) \wedge (\mu_P(b_1) > \mu_P(b_3))$,
then LN$_{(n,2)}(P) = (2b_1 + b_2)/3$.
Otherwise, if $b_2 > M_{b_1b_3}$ or
$(b_2 = M_{b_1b_3}) \wedge (\mu_P(b_1) < \mu_P(b_3))$,
then LN$_{(n,2)}(P) = (b_2 + 2b_3)/3$.

\item
When $P$ is type W:
$k_P = 1$ and $m_P = 4$,
and $P$ satisfies either condition (a) or (b) (of the definition of type W).
\begin{description}
\item[(a)] 
If $2(b_2 - b_1) = b_3 - b_2$ and $b_3 \leq M_{b_1b_4}$,
then LN$_{(n,2)}(P) = b_2$.

\item[(b)]
If $2(b_4 - b_3) = b_3 - b_2$ and $b_2 \geq M_{b_1b_4}$,
then LN$_{(n,2)}(P) = b_3$.
\end{description}

\item
When $P$ is type U$_4$:
$k_P = 1$, $m_P = 4$, and $P$ is not type W.
Suppose that $\mu_P(b_1) \geq \mu_P(b_4)$ holds.
(The case $P$ satisfies $\mu_P(b_1) < \mu_P(b_4)$ is symmetric, 
and we omit it.)

\begin{description}
\item[(a)] 
If $\mu_P(b_1) \geq \mu_P(b_3)$,
then LN$_{(n,2)}(P) = b_1$.

\item[(b)]
If $(\mu_P(b_1) < \mu_P(b_3)) \wedge (\mu_P(b_3) \geq 3)$,
LN$_{(n,2)}(P) = b_1$, if $b_3 = 0$, 
and LN$_{(n,2)}(P) = 0$, otherwise if $b_3 \not= 0$.

\item[(c)]
Otherwise if $(\mu_P(b_1) < \mu_P(b_3)) \wedge (\mu_P(b_3) < 3)$,
$\mu_P(b_1) = \mu_P(b_4) = 1$ and $\mu_P(b_3) = 2$.
LN$_{(n,2)}(P) = b_1$, if $(b_2 = 0) \vee (b_3 = 0)$, 
and LN$_{(n,2)}(P) = 0$, otherwise if $(b_1 = 0) \vee (b_4 = 0)$.
\end{description}
\end{enumerate}

\begin{lemma}
\label{L4050}
Target function LN$_{(n,2)}$ is an algorithm to solve the \mbox{\rm FC(2)-PO}
for any configuration of type L, without reaching a configuration
not in type L.
\end{lemma}

\begin{proof}
Our proof scenario is similar to the proof of Lemma~\ref{L4040}.
By an exhaustive search,
we draw the transition diagram among the types,
and show that if the execution does not converge to at most two points,
it eventually reaches a type G configuration.

We borrow symbols and notations from the proof of Lemma~\ref{L4040}.

Consider any execution ${\cal E}: P_0, P_1, \ldots$ starting from any
initial configuration $P_0$ of type L.
Note that $P_0$ (in $Z_0$) is type L, but may not contain $(0,0)$.
A configuration $Q_0$ identified by a robot at time $0$, however,
is type L and contains $(0,0)$.
Then $Q_0$, after identified with a set of real numbers including 0, 
applies to LN$_{(n,2)}$.
By the definition of LN$_{(n,2)}$,
$CH_{t+1} \subseteq CH_t$,
which implies that $P_t$ is type L for all $t \geq 0$.
Like the proof of Lemma~\ref{L4040},
we make a series of observations (I)--(XV).
Figure~\ref{F0030} shows the transition diagram among the types,
which summarizes the observations (I)--(V), (VII), and (X)--(XIV).

\begin{figure}[t]
\centering
\includegraphics[width=\hsize]{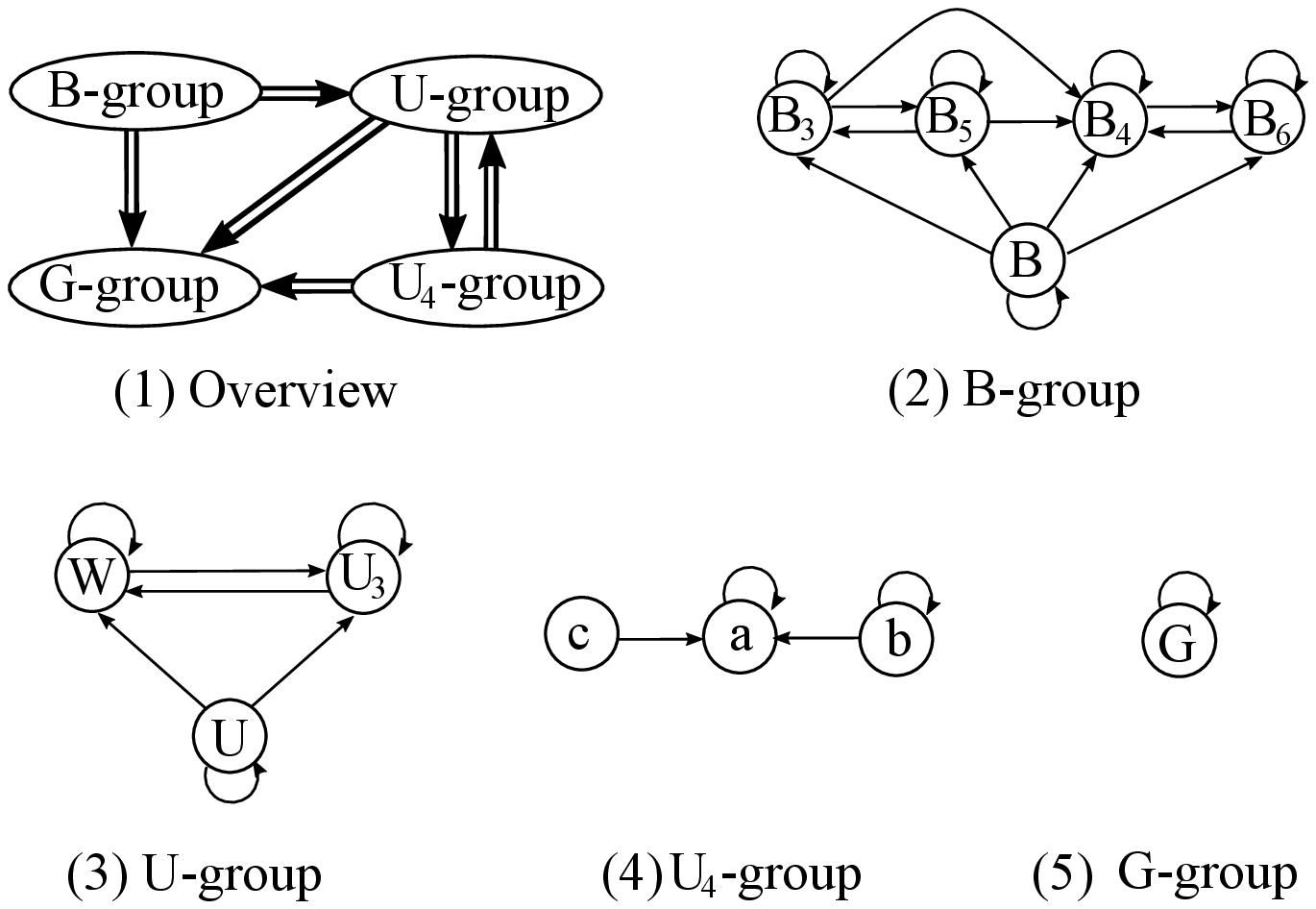}
\caption{The transition diagram among the types that LN$_{(n,2)}$ specifies.
Part (1): The overview of the transition diagram. 
Part (2): B-group contains types B, B$_3$, B$_4$, B$_5$, and B$_6$.  
Part (3): U-group contains types U, U$_3$, and W.
Part (4): U$_4$-group contains three cases (a), (b), and (c) in type U$_4$.
Part (5): G-group contains type G.
In each of Parts (2)-(5), 
an arrow represents a transition between two types 
(or two cases in U$_4$-group).
In Part (1), a double arrow from a group to another group represents that
there is a transition from a type in the former group to a type in the latter one.
A loop between U-group and U$_4$-group in Part (1)
eventually terminates by (XI) and (XV).
}
\label{F0030}
\end{figure}

\begin{description}
\item[(I) When $P_t$ is type G.] 

Any type G configuration satisfies the goal condition,
i.e., all robots have already gathered at a point or two.
Since $\mathrm{LN}_{(n,2)}(P_t) = 0$,\footnote{
Formally, this equation does not hold.
See the last footnote.}
the target point is the current position,
hence the execution stays type G forever,
once it reaches a configuration of type G.

\item[(II) When $P_t$ is type B.]

Suppose that $P_t$ is type B.
Then $P_{t+1}$ can be any type.

If $P_{t+1}$ is type B,
each activated robot, as long as it is not faulty, 
moves either to $b_1$ or $b_{m_t}$ at time $t$.
Unless $m_{t+1} < m_t$ (i.e., if $m_{t+1} = m_t$),
$\mu_{P_{t+1}}(b_1) + \mu_{P_{t+1}}(b_{m_t}) > \mu_{P_t}(b_1) + \mu_{P_t}(b_{m_t})$.
Thus this self-loop from B to B cannot repeat more than $n$ times,
and the execution eventually reaches a configuration
whose type is not B.

\item[(III) When $P_t$ is type B$_3$.]

Suppose that $P_t$ is type B$_3$.
Then the type of $P_{t+1}$ is either G, B$_3$, B$_4$, B$_5$, U$_3$, U$_4$, or U.

Let $L_t = L_{P_t}$, which is the length of $CH_t$.
If $P_{t+1}$ is type B$_3$, then $L_{t+1} = L_t/2$.
Thus the execution converges to a point,
if this self-loop from B$_3$ to B$_3$ repeats infinitely many times.

\item[(IV) When $P_t$ is type B$_4$.]

Suppose that $P_t$ is type B$_4$.
Then the type of $P_{t+1}$ is either G, B$_4$, B$_6$, U$_3$, U$_4$, or U.

Let $\lambda_t = \lambda_{P_t}$.
If $P_{t+1}$ is type B$_4$, $\lambda_{t+1} = \lambda_t/2$.
Thus the execution converges to two points,
if this self-loop from B$_4$ to B$_4$ repeats infinitely many times.

\item[(V) When $P_t$ is type B$_5$.]

Suppose that $P_t$ is type B$_5$.
Then the type of $P_{t+1}$ is either G, B$_3$, B$_4$, B$_5$, U$_3$, U$_4$, or U.

If $P_{t+1}$ is type B$_5$,
each activated robot, as long as it is not faulty, 
moves either to $b_2$ or $b_4$ at time $t$.
Unless $m_{t+1} < m_t$ (i.e., if $m_{t+1} = m_t$),
$\mu_{P_{t+1}}(b_2) + \mu_{P_{t+1}}(b_4) > \mu_{P_t}(b_2) + \mu_{P_t}(b_4)$.
Thus this self-loop from B$_5$ to B$_5$ cannot repeat more than $n$ times,
and the execution eventually reaches a configuration
whose type is not B$_5$.

\item[(VI) B$_3$B$_5^+$B$_3$ loop.]

Suppose that $P_t$ is type B$_3$.
Then as observed in (III),
$P_{t+1}$ can be type B$_5$.
Then as observed in (V),
after several repetition of the self-loop of B$_5$,
the execution can reach a configuration $P_{t'}$ of type $B_3$ 
for the first time after $t$.
Let $\overline{P_t} = \{ a,b,c \}$, where $b - a = c - b$.
Then $\overline{P_{t+1}} = \{ a, M_{ab}, b, M_{bc}, c \}$,
and $\overline{P_{t'}} = \{ M_{ab}, b, M_{bc}\}$.
Thus $L_{t'} = L_t/2$.
If this B$_3$B$_5^+$B$_3$ loop repeats infinitely many times,
the execution converges to a point.

\item[(VII) When $P_t$ is type B$_6$.]

Suppose that $P_t$ is type B$_6$.
Then the type of $P_{t+1}$ is either G, B$_4$, B$_6$, U$_3$, W, U$_4$, or U.

If $P_{t+1}$ is type B$_6$,
each activated robot, as long as it is not faulty, 
moves either to $b_2$ or $b_5$ at time $t$.
Unless $m_{t+1} < m_t$ (i.e., if $m_{t+1} = m_t$),
$\mu_{P_{t+1}}(b_2) + \mu_{P_{t+1}}(b_5) > \mu_{P_t}(b_2) + \mu_{P_t}(b_5)$.
Thus this self-loop from B$_6$ to B$_6$ cannot repeat more than $n$ times,
and the execution eventually reaches a configuration whose type is not B$_6$.

\item[(VIII) B$_4$B$_6^+$B$_4$ loop.]

Suppose that $P_t$ is type B$_4$.
Then as observed in (IV),
$P_{t+1}$ can be type B$_6$.
Then as observed in (VII),
after several repetition of the self-loop of B$_6$,
the execution can reach a configuration $P_{t'}$ of type $B_4$ 
for the first time after $t$.
Let $\overline{P_t} = \{ a,b,c,d \}$, 
where $a<b<c<d$ and $b - a = d - c$.
Then $\overline{P_{t+1}} = \{ a, M_{ab}, b, c, M_{cd}, d \}$,
and $\overline{P_{t'}} = \{ x, M_{ab}, M_{cd}, y\}$,
where $x \in \{a, b\}$ and $y \in \{c, d\}$.
(We ignore the order among $x, M_{ab}, M_{cd}, y$ in $\overline{P_{t'}}$.)
Thus $\lambda_{t'} = \lambda_t/2$.
If this B$_4$B$_6^+$B$_4$ loop repeats infinitely many times,
the execution converges to two points.

\item[(IX) Summary of (II)--(VIII).]

Suppose that $k_0 = k_{P_0} = 2$.
Unless the execution converges to a point or two,
by repeating B$_3$B$_5^+$B$_3$ loop or B$_4$B$_6^+$B$_4$ loop infinitely many times,
it eventually reaches a type G configuration or
a configuration $P_t$ such that $k_t = 1$.

\item[(X) When $P_t$ is type U.]

Suppose that $P_t$ is type U.
Then the type of $P_{t+1}$ is either G, U$_3$, W, U$_4$, or U.

If $b_1 \succ_{P_t} b_{m_t}$,
then $\mu_{P_t}(b_1) \geq \mu_{P_t}(b_{m_t})$.
Since LN$_{(n,2)}(P_t) = b_1$,
unless $m_{t+1} < m_t$,
$\mu_{P_{t+1}}(b_1) > \mu_{P_t}(b_1) \geq \mu_{P_t}(b_{m_t}) \geq \mu_{P_{t+1}}(b_{m_t})$.
Thus $b_1 \succ_{P_{t+1}} b_{m_t}$ at $t + 1$, and $P_{t+1}$ is type U.
This self-loop from U to U cannot repeat more than $n$ times.

Otherwise, if $b_{m_t} \succ b_1$,
similarly, $P_{t+1}$ is type U,
and this self-loop from U to U cannot repeat more than $n$ times.

Thus the execution eventually reaches a configuration of type G, U$_3$, W, or U$_4$. 

\item[(XI) When $P_t$ is type U$_3$.]

Suppose that $P_t$ is type U$_3$.
Then $k_{t+1} = 1$ unless the type of $P_{t+1}$ is G.
The type of $P_{t+1}$ is either G, U$_3$, or W (but not U$_4$).

Suppose that the type of $P_{t+1}$ is U$_3$.
Then all the robots at exactly one of $b_1, b_2$, and $b_3$
have moved to LN$_{(n,2)}(P_t)$ at $t$,
which implies that $\lambda_{t+1} \leq 2\lambda_t/3$.
Thus the execution eventually converges to two points,
if this self-loop from U$_3$ to U$_3$ repeats infinitely many times.

\item[(XII) When $P_t$ is type W.]

Suppose that $P_t$ is type W.
By the definition of LN$_{(n,2)}$,
$k_{t+1} = 1$, 
and the type of $P_{t+1}$ is either G, U$_3$, or W.
Furthermore, if $P_{t+1}$ is type W,
then $\lambda_{t+1} = \lambda_t$ holds.

Suppose that $b_3 \leq M_{b_1b_4}$.
Since $\mu_{P_{t+1}}(b_2) > \mu_{P_t}(b_2)$,
this self-loop from W to W can repeat at most $n$ times,
and the execution eventually reaches a configuration of type G or U$_3$.

Suppose otherwise that $b_2 \geq M_{b_1b_4}$.
Then $\mu_{P_{t+1}}(b_3) > \mu_{P_t}(b_3)$,
this self-loop from W to W can repeat at most $n$ times,
and the execution eventually reaches a configuration of type G or U$_3$.

\item[(XIII) U$_3$W$^+$U$_3$ loop.]

Suppose that the type of $P_t$ is U$_3$,
and that of $P_{t+1}$ is W.
As observed in (XII),
the execution can reach a configuration $P_{t'}$ of type U$_3$ 
for the first time after $t$.
Then by (XI) and (XII), $\lambda_{t'} \leq 2\lambda_t/3$.
Thus the execution converges to two points,
if this U$_3$W$^+$U$_3$ loop repeats infinitely many times.

\item[(XIV) When $P_t$ is type U$_4$.]

Suppose that $P_t$ is type U$_4$ and satisfies $\mu_{P_t}(b_1) \geq \mu_{P_t}(b_4)$.
We consider three cases corresponding to the three cases (a)-(c)
in the definition of LN$_{(n,2)}$.

\begin{description}
\item[(a) When $\mu_{P_t}(b_1) \geq \mu_{P_t}(b_3)$:]
Since LN$_{(n,2)}(P_t) = b_1$,
$\mu_{P_{t+1}}(b_1) > \mu_{P_t}(b_1) \geq \mu_{P_t}(b_4) \geq \mu_{P_{t+1}}(b_4)$.
Thus $k_{t+1} = 1$, $m_{t+1} \leq 4$, 
and the type of $P_{t+1}$ is either G, U$_3$, or U$_4$.

Moreover, if $P_{t+1}$ is type U$_4$, then it satisfies the condition (a).
Then this self-loop from U$_4$(a) to U$_4$(a) cannot repeat more than $n$ times,
and the execution eventually reaches a configuration of type G or U$_3$.

\item[(b) When $(\mu_{P_t}(b_1) < \mu_{P_t}(b_3)) \wedge (\mu_{P_t}(b_3) \geq 3)$:]
By the definition of LN$_{(n,2)}$,
only robots at $b_3$ can move (to $b_1$),
and the other robots cannot move (even if they are activated),
since LN$_{(n,2)}(P_t) = 0$.

First confirm that at least one robot $r$ at $b_3$ is non-faulty,
and eventually $r$ is activated to change the configuration,
since there are at most two faulty robots.
Thus $k_{t+1} = 1$, $m_t \leq 4$,
and the type of $P_{t+1}$ is either U$_3$ or U$_4$.

If $P_{t+1}$ is type U$_4$, then it satisfies condition (a) or (b).
If $P_{t+1}$ satisfies condition (b),
this self-loop from U$_4$(b) to U$_4$(b) cannot repeat more than $n$ times,
since $\mu_{P_{t+1}}(b_1) > \mu_{P_t}(b_1) \geq \mu_{P_t}(b_4) \geq \mu_{P_{t+1}}(b_4)$,
and the execution eventually reaches a configuration of type G or U$_3$.

On the other hand, if $P_{t+1}$ satisfies condition (a),
then as observed in case (a), 
the execution eventually reaches a configuration of type G or U$_3$.

Thus the execution eventually reaches a configuration of type G or U$_3$,
regardless of whether or not $P_{t+1}$ satisfies condition $(b)$.

\item[(c) When $(\mu_{P_t}(b_1) < \mu_{P_t}(b_3)) \wedge (\mu_{P_t}(b_3) < 3)$:]
Condition (c) holds, if and only if 
$1 \leq \mu_{P_t}(b_4) = \mu_{P_t}(b_1) < \mu_{P_t}(b_3) = 2$ 
(and $\mu_{P_t}(b_2) \geq 1$) hold.
Since $\mu_{P_t}(b_2) + \mu_{P_t}(b_3) \geq 3$,
there is at least one non-faulty robot $r$ at $b_2$ or $b_3$,
and $r$ eventually moves to $b_1$, 
since LN$_{(n,2)}(P_t) = b_1$, if $(b_2 = 0) \vee (b_3 = 0)$.
Since $\mu_{P_{t+1}}(b_1) \neq \mu_{P_{t+1}}(b_4)$, $k_{t+1} = 1$, 
and $m_{t+1} \leq 4$. 
Then the type of $P_{t+1}$ is G, U$_3$, or U$_4$.

Moreover if $P_{t+1}$ is type U$_4$, 
then it must satisfy condition (a),
since $\mu_{t+1}(b_1) \geq 2 \geq \mu_t(b_3) \geq \mu_{t+1}(b_3)$.
\end{description}

Let us summarize:
When $P_t$ is type U$_4$,
eventually the execution reaches a configuration of type G or U$_3$.

\item[(XV) Summary of (X)--(XIV).]

Suppose that $k_0 = k_{P_0} = 1$.
Unless the execution converges to two points
by repeating U$_3$ self-loop or U$_3$W$^+$U$_3$ loop infinitely many times,
it eventually reaches a type G configuration.
That is, the loop between U-group and U$_4$-group in Overview of Figure ~\ref{F0030}
eventually terminates, by (XI) and the summary of (XIV).
\end{description}

Now, we go on the proof.
The observations summarized in Figure~\ref{F0030} show that 
any execution $\mathcal{E}$ starting from a configuration of any type eventually 
reaches a type G configuration, 
if neither B$_3$B$_5^+$B$_3$-loop, B$_4$B$_6^+$B$_4$-loop, nor
U$_3$W$^+$U$_3$-loop repeats infinitely many times.
We can conclude the correctness of LN$_{(n,2)}$ by observations (IX) and (XV)
which show that any one of these loops cannot repeat infinitely many times 
without reaching a type G configuration.
\end{proof}

\medskip

It is easy to see that $\alpha(\psi_{(n,2)})=1$.
By Lemmas~\ref{L4040} and \ref{L4050},
we have the following theorem.

\begin{theorem}
\label{T4050}
Target function $\psi_{(n,2)}$,
which satisfies $\alpha(\psi_{(n,2)})=1$, is an algorithm 
for the \mbox{\rm FC(2)-PO}.
\end{theorem}

\section{Gathering problem}
\label{Sgather}

We finally investigate the gathering problem,
provided that there are no faulty robots,
to emphasize that the gathering and the convergence problems have
completely different properties from the viewpoint of compatibility.

Since the gathering problem is not solvable if $n = 2$ \cite{SY99},
we assume $n \geq 3$ in this section.
Moreover, we assume that the robots initially occupy distinct points,
as all proposed gathering algorithms assume.
There are many gathering algorithms.
The following algorithm GAT \cite{SY99} is one of them.

Let $P = \{ \bm{p}_1, \bm{p}_2, \ldots , \bm{p}_n \}$ be a multiset 
of $n \geq 3$ points.
We use concepts $\overline{P}, \mu_P, \bm{o}_P, k_P$, and $\succ_P$ 
introduced in Subsection~\ref{SSSAn2}.

\medskip
\noindent
{\bf [Target function GAT]}
\begin{enumerate}
\item
If there is a unique $\bm{p} \in P$ such that $\mu_P(\bm{p}) > 1$,
GAT$(P) = \bm{p}$.

\item
Otherwise, if $\mu_P(\bm{p}) = 1$ for all $\bm{p} \in P$:

\begin{description}
\item[(a)]
If $k_P \leq 1$,
GAT$(P) = \bm{p}$,
where $\bm{p}$ is the largest point in $P$ with respect to $\succ_P$.

\item[(b)]
If $k_P > 1$,
GAT$(P) = \bm{o}_P$.
\end{description}
\end{enumerate}

\noindent
Observe that $\alpha({\rm GAT}) = 1$.
A sketch of the correctness proof of GAT is as follows:
Since the robots initially occupy distinct points by assumption,
and by the definition of Step~2 of GAT,
in any execution of GAT,
there must be a time such that a unique point $\bm{q}$\footnote{
Here $\bm{q}$ is either $\bm{p}$ or $\bm{o}_P$ in Step~2 of GAT.} 
satisfying $\mu(\bm{q}) > 1$ occurs in the configuration for the first time.
Then by the definition of Step~1,
$\mu(\bm{q})$ monotonically increases
(while $\mu(\bm{x})$ of the other points $\bm{x} \in P$ monotonically decreases),
and eventually $\mu(\bm{q}) = n$ is satisfied.

We can modify Step~2(a) of GAT to obtain another algorithm GAT$'$.
For example, GAT$'(P)$ can be the smallest point $\bm{p}'$ in $P$ 
with respect to $\succ_P$, instead of the largest point $\bm{p}$.
Then indeed GAT$'$ is also a gathering algorithm with $\alpha({\rm GAT'}) = 1$,
but obviously $\Phi = \{{\rm GAT, GAT'} \}$ is not compatible 
with respect to the gathering problem;
if some robots take GAT and some others GAT$'$,
then a configuration $P$ such that $\mu_P(\bm{p}), \mu_P(\bm{p}') \geq 2$ may yield.
Let us summarize.

\begin{theorem}~{\rm \cite{SY99}}
\label{T5005}
Let $\Phi = \{ {\rm GAT} \}$ and $\Phi' = \{ {\rm GAT'} \}$.
Then $\Phi$ and $\Phi'$ are compatible with respect to the gathering problem,
but $\Phi \cup \Phi'$ is not.

Here $\alpha(\Phi) = \alpha(\Phi') = \alpha(\Phi \cup \Phi') = 1$.
\end{theorem}

\begin{theorem}
\label{T5010} 
Any target function $\phi$ is not a gathering algorithm
if $\alpha(\phi) < 1$, or equivalently,
any set $\Phi$ of target functions such that $\alpha(\Phi) < 1$ is not
compatible with respect to the gathering problem.
\end{theorem}

\begin{proof}
Consider any gathering algorithm $\phi$
and show that $\alpha(\phi) = 1$.

Suppose that $n = 3$.
For any initial configuration $P_0 $ satisfying that all robots occupy distinct positions,
we investigate any execution ${\cal E}: P_0, P_1, \ldots$,
assuming that the scheduler is central,
i.e., exactly one robot is activated each time.

Since ${\cal E}$ achieves the gathering,
there is a time instant $t$ such that $|\overline{P_t}| = 2$ 
and $|\overline{P_{t+1}}| = 1$.
Since exactly one robot, say $r$, is activated at $t$,
it moves to the position of the other robots (since they occupy the same position).
Hence $\alpha(\phi) = 1$.
\end{proof}

\section{Conclusions}
\label{Sconclusion}

We introduced the concept of compatibility 
and investigated the compatibilities of several convergence problems.
A compatible set $\Phi$ of target functions with respect to a problem $\Pi$ 
is an extension of an algorithm $\phi$ for $\Pi$,
in the sense that every target function $\phi \in \Phi$ is an algorithm for $\Phi$,
although a set of algorithms for $\Pi$ is not always a compatible set with respect to $\Pi$.

The problems we investigated are the convergence problem,
the fault tolerant $(n,f)$-convergence problem (FC($f$)),
the fault tolerant $(n,f)$-convergence problem to a convex $f$-gon (FC($f$)-CP), and
the fault tolerant $(n,f)$-convergence problem to $f$ points (FC($f$)-PO),
for crash faults.
The gathering problem was also investigated.
The results are summarized in Table~\ref{Table0010}.
Main observations we would like to emphasize are:
\begin{enumerate}
 \item 
The convergence problem, the FC$(1)$, the FC$(1)$-PO, and the FC($f$)-CP 
share the same property:
Every set $\Phi$ of target functions such that $0\leq \alpha(\Phi) < 1$ is compatible.
\item
The gathering problem and the FC$(f)$-PO for $f \geq 2$ share the same property:
Any set $\Phi$ of target functions such that $0\leq \alpha(\Phi) < 1$ is {\bf not} compatible.
\item
FC($f$) ($f \geq 2$) is in between the FC($f$)-CP and the FC($f$)-PO.
\item
The FC(1)-PO and the FC(2)-PO are completely different problems 
from the viewpoint of compatibility.
\end{enumerate}

A compatible set $\Phi$ with respect to a problem $\Pi$ could be regarded
as an ``algorithm scheme'' $\cal S$ for $\Pi$ such that every algorithm 
in $\Phi$ is an instantiation of $\cal S$.
For example, for any fixed $0 \leq \alpha < 1$, 
the set $\Phi_{\alpha} = \{ \phi: \alpha(\phi) \leq \alpha \}$ 
of target functions $\phi$ is compatible with respect to the convergence problem.
Then an algorithm scheme $\cal S$ defining $\Phi_{\alpha}$ 
(as the set of its instantiations) is 
\begin{quote}
$\phi(P) \in \alpha*CH(P)$ for all $P \in {\cal P}$.
\end{quote}
It is not an algorithm, 
since it does not specify a concrete value of $\phi(P)$.
It however captures the essence how all algorithms in common 
solve the convergence problem,
and we can show the correctness of each algorithm from this description.

For any $0 < \delta \leq 1/2$,
the set $\Lambda_{\delta} = \{ \phi: \phi ~\mbox{\rm is $\delta$-inner} \}$ 
is compatible with respect to the convergence problem \cite{CDFH11}.
For any $0 \leq \alpha < 1$,
there is a $0 < \delta \leq 1/2$ such that
$\Phi_{\alpha} \subset \Lambda_{\delta}$,
which implies that there is another algorithm scheme for the convergence problem
which is more general than algorithm scheme: 
$\phi(P) \in \alpha*CH(P)$ for all $P \in {\cal P}$.
Investigation of an algorithm scheme would lead us a deeper
understanding of robot algorithms.

Note that the concept of algorithm scheme is not new:
In Chapter 21 of \cite{CLRS22}, for example, 
the authors first present GENERIC-MST,
an algorithm scheme to construct a minimum spanning tree
(and show its correctness),
and then derive Kruscal's and Prim's algorithms as its instantiations.

Before closing the paper,
we list some open problems:
\begin{enumerate}
\item
Extend Table~\ref{Table0010} to contain the results for $\alpha(\Phi) > 1$.

\item
Suppose that $\phi$ and $\phi'$ are algorithms for the convergence problem.
Find a necessary and/or a sufficient condition for $\Phi = \{ \phi, \phi' \}$ 
to be compatible with respect to the convergence problem.

More generally,
if a set $\Phi$ of target functions is compatible with respect to the convergence problem,
then every target function $\phi \in \Phi$ is a convergence algorithm.
What is a sufficient condition for a set of convergence algorithms
to be compatible with respect to the convergence algorithm?

\item
Is there an algorithm for the FC$(3)$-PO?

\item
Is the next statement correct?
If a set $\Phi$ of target functions is compatible with respect to the convergence problem,
so is $\beta \Phi = \{ \beta \phi : \phi \in \Phi \}$ for any real number $0 < \beta \leq 1$.

\item
Let $ALG_{{\rm FC(1)}}$ (resp. $ALG_{\textrm{FC(1)-PO}}$) 
be the set of all algorithms for the FC(1) (resp. the FC(1)-PO).
Does $ALG_{\textrm{FC(1)-PO}} = ALG_{{\rm FC(1)}}$ hold?

\item
Investigate the compatibility of convergence problems under the $\cal ASYNC$ model.

\item
Investigate the compatibility of convergence problems in the presence of Byzantine failures.

\item
Investigate the compatibility of fault tolerant gathering problems.

\item
Find interesting problems with a large compatible set.

\end{enumerate}

\end{document}